\documentclass{article}
\usepackage{fullpage}
\usepackage{epsfig}
\usepackage{amsthm}
\usepackage{amsmath}
\usepackage{url} 

\newcommand{\comment}[1]{{}}
\newcommand{\old}[1]{{}}

\newcommand{\B}{BB({\cal R})}
\newcommand{\eps}{\varepsilon}

\newtheorem{theorem}{Theorem}[section]
    \newtheorem{lemma}[theorem]{Lemma}
    
\newtheorem{claim}[theorem]{Claim}

\newtheorem{observation}{Observation}

\newtheorem{conjecture}{Conjecture}

\title{Approximating Maximum Independent Set for Rectangles in the Plane}

\author{Joseph S. B. Mitchell\thanks{Stony Brook University, 
Stony Brook, NY 11794-3600, joseph.mitchell@stonybrook.edu}}

\begin{document}

\maketitle

\begin{abstract}
We give a polynomial-time constant-factor approximation algorithm for
maximum independent set for (axis-aligned) rectangles in the plane.
Using a polynomial-time algorithm, the best approximation factor
previously known is $O(\log\log n)$.  The results are based on a new
form of recursive partitioning in the plane, in which faces that are
constant-complexity and orthogonally convex are recursively
partitioned into a constant number of such faces.
\end{abstract}

\section{Introduction}

Given a set ${\cal R}=\{R_1,\ldots,R_n\}$ of $n$ axis-aligned
rectangles in the plane, the {\em maximum independent set of
  rectangles} (MISR) problem seeks a maximum-cardinality subset,
${\cal R}^*\subseteq {\cal R}$, of rectangles that are {\em
  independent}, meaning that any two rectangles of ${\cal R}^*$ are
interior-disjoint.

The MISR is an NP-hard (\cite{fowler1981optimal,imai1983finding})
special case of the fundamental Maximum Independent Set (MIS)
optimization problem, in which, for a given input graph, one seeks a
maximum-cardinality subset of its vertices such that no two vertices
are joined by an edge.
The MIS is an optimization problem that is very difficult to
approximate, in its general form; the best polynomial-time
approximation algorithm is an $O(n/\log^2
n)$-approximation~\cite{boppana1992approximating}, and there is no
polynomial-time $n^{1-\delta}$-approximation algorithm, for any fixed
$\delta>0$, unless P$=$NP~\cite{zuckerman2007linear}.

Because the MIS is such a challenging problem in its general setting,
the MIS in special settings, particularly geometric settings, has
been extensively studied. For string graphs (arcwise-connected sets), an
$n^\eps$-approximation algorithm is known~\cite{fox2011computing}.
For outerstring graphs, a polynomial-time exact algorithm is
known~\cite{outerstring}, provided a geometric representation of the
graph is given.
For geometric objects that are ``fat''(e.g., disks, squares), a
Polynomial-Time Approximation Scheme (PTAS) is
known~\cite{erlebach2005polynomial}.  Fatness is critical in many
geometric approximation settings.  The MISR is one of the most natural
geometric settings for MIS in which the objects are {\em not} assumed
to be fat -- rectangles can have arbitrary aspect ratios.
The MISR arises naturally in various applications, including
resource allocation, 
data mining, 
and map labeling. 

For the MISR, the celebrated results of Adamaszek, Har-Peled, and
Wiese~\cite{adamaszek2019approximation,adamaszek2013approximation,adamaszek2014qptas,har2014quasi}
give a {\em quasipolynomial}-time approximation scheme, yielding a
$(1-\eps)$-approximation in time $n^{poly(\log n/\eps)}$.
Chuzhoy and Ene~\cite{chuzhoy2016approximating} give an improved, but
still not polynomial, running time: they present a
$(1-\eps)$-approximation algorithm with a running time of
$n^{O((\log\log n/\eps)^4)}$
%
%
Recently, Grandoni, Kratsch, and
Wiese~\cite{grandoni2019parameterized} give a parameterized
approximation scheme, which, for given $\eps>0$ and integer $k>0$,
takes time $f(k,\eps)n^{g(\eps)}$ and either outputs a solution of
size $\geq k/(1+\eps)$ or determines that the optimum solution has
size $<k$.

Turning now to {\em polynomial}-time approximation algorithms for
MISR, early algorithms achieved an approximation ratio of $O(\log
n)$~\cite{agarwal1998label,khanna1998approximating,nielsen2000fast} and $\lceil \log_k n\rceil$, for any
fixed $k$~\cite{berman2001improved}.  In a significant breakthrough, the $O(\log
n)$ approximation bound was improved to $O(\log\log
n)$~\cite{chalermsook2009maximum,Chalermsook-APPROX11}, which has
remained the best approximation factor during the last decade.
Indeed, it has been a well known open problem to improve on this
approximation factor: ``even for the axis-parallel rectangles case
currently has no constant factor approximation algorithm in polynomial
time''~\cite{adamaszek2019approximation}, and ``obtaining a PTAS, or
even an efficient constant-factor approximation remains elusive for
now''~\cite{chuzhoy2016approximating}.

\subsection*{Our Contribution}

We give a polynomial-time $O(1)$-approximation algorithm for MISR,
improving substantially on the prior $O(\log\log n)$ approximation
factor for a polynomial-time algorithm. The algorithm is based on a
novel means of recursively decomposing the plane, into
constant-complexity (orthogonally convex) ``corner-clipped
rectangles'' (CCRs), such that any set of disjoint (axis-aligned)
rectangles has a constant-fraction subset that ``nearly respects'' the
decomposition, in a manner made precise below.
%
%
The fact that the CCRs are of constant complexity, and the
partitioning recursively splits a CCR into a constant number of CCRs,
allow us to give an algorithm, based on dynamic programming, that
computes such a subdivision, while maximizing the size of the subset
of ${\cal R}$ of rectangles that are not penetrated (in their
interiors) by the cuts. Together, the structure theorem and the
dynamic program, imply the result.

Two key ideas are used in obtaining our result: (1) instead of a
recursive partition based on a binary tree, partitioning each problem
into two subproblems, we utilize a $K$-ary partition, with constant
$K\leq 3$; (2) instead of subproblems based on rectangles (arising
from axis-parallel cuts), we allow more general cuts, which are
constant-complexity planar graphs having edges defined by
axis-parallel segments, with resulting regions being orthogonally
convex (CCR) faces of constant complexity.  Proof of our structural
theorem relies on a key idea to use maximal expansions of an input set
of rectangles, enabling us to classify rectangles of an optimal
solution in terms of their ``nesting'' properties with respect to
abutting maximal rectangles. This sets up a charging scheme that can
be applied to a subset of at least half of the maximal expansions of
rectangles in an optimal solution: those portions of cuts that cross
these rectangles can be accounted for by means of a charging scheme,
allowing us to argue that a constant fraction subset of original
rectangles exists that (nearly) respects a recursive partitioning into
CCRs. These properties then enable a dynamic programming algorithm to
achieve a constant factor approximation in polynomial time.

\section{Preliminaries}

Throughout this paper, when we refer to a rectangle we mean an
{\em axis-aligned}, closed rectangle in the plane.
We let $[x,x']\times[y,y']$ denote the rectangle whose projection on
the $x$-axis (resp., $y$-axis) is the (closed) interval $[x,x']$
(resp., $[y,y']$).
We often speak of the {\em sides} (or {\em edges}) of a rectangle; we
consider each rectangle to have four sides (top, bottom, left, right),
each of which is a closed (horizontal or vertical) line segment. Two
nonparallel sides of a rectangle meet at a common endpoint, one of the
{\em corners} of the rectangle, often distinguished as being
northeast, northwest, southwest, or southeast.
When we speak of rectangles being disjoint, we mean that their
interiors are disjoint.
We say that two rectangles {\em abut} (or {\em are in contact}) if
their interiors are disjoint, but there is an edge of one rectangle
that has a nonzero-length overlap with an edge of the other rectangle.
We say that a horizontal/vertical line segment $\sigma$ {\em
  penetrates} a rectangle $R\in{\cal R}$ if $\sigma$ intersects the
interior of $R$.  We say that a horizontal/vertical line segment
$\sigma$ {\em crosses} another horizontal/vertical segment $\sigma'$
if the two segments intersect in a point that is interior to both
segments; similarly, we say that $\sigma$ {\em crosses} a rectangle
$R$ if the intersection $\sigma\cap R$ is a subsegment of $\sigma$
that is contained in the interior of~$\sigma$.

The input to our problem is a set ${\cal R}$ of $n$ arbitrarily
overlapping (axis-aligned) rectangles in the plane.  Without loss of
generality, we can assume that the $x$-coordinates and $y$-coordinates
that define the rectangles ${\cal R}$ are distinct and are
integral. Let $x_1,x_2,\ldots,x_{2n}$, with $x_1< x_2< \cdots <
x_{2n}$, be the sorted $x$-coordinates of the left/right sides of the
$n$ input rectangles ${\cal R}$; similarly, let $y_1< y_2 < \cdots <
y_{2n}$ be the $y$-coordinates of the top/bottom sides of the input.
We let $\B=[x_1,x_{2n}]\times[y_1,y_{2n}]$ denote the axis-aligned
bounding box (minimal enclosing rectangle) of the rectangles ${\cal
  R}$.  The arrangement of the horizontal/vertical lines ($x=x_i$,
$y=y_j$) that pass through the sides of the rectangles ${\cal R}$
determine a {\em grid} ${\cal G}$, with faces (rectangular grid
cells), edges (line segments along the defining lines), and vertices
(at points where the defining lines cross).
%

In the maximum independent set of rectangles (MISR) problem on ${\cal
  R}$ we seek a subset, ${\cal R}^*\subseteq {\cal R}$, of maximum
cardinality, $k^*=|{\cal R}^*|$, such that the rectangles in ${\cal
  R}^*$ are {\em independent} in that they have pairwise disjoint
interiors. Our main result is a polynomial-time algorithm to compute
an independent subset of ${\cal R}$ of cardinality $\Omega(k^*)$.

\medskip
\noindent{\bf Corner-Clipped Rectangles (CCRs).}\quad
An {\em orthogonal polygon} is a simple polygon\footnote{We consider
  rectangles and polygons to be closed regions that include their
  boundary and their interior.} whose boundary consists of a finite
set of axis-parallel segments.  A {\em corner-clipped rectangle} $Q$
is a special orthogonal polygon obtained from a given rectangle $R$ by subtracting from $R$ a set of up to
four interior-disjoint rectangles, with each
rectangle containing exactly one of the four corners of $R$ in its interior.
In an earlier version of this
paper~\cite{DBLP:journals/corr/abs-2101-00326}, we utilized the full
generality of a CCR, having potentially all four corners clipped (resulting in a CCR with
up to 12 sides).  In this paper we employ a simplified case analysis, and we restrict ourselves to CCRs
with at most {\em one} corner clipped, so our CCRs are 
``L-shaped'' orthogonal polygons with 6 sides, 1 reflex vertex (of
internal angle 270 degrees), and 5 convex vertices (each with internal
angle 90 degrees).  Refer to Figure~\ref{fig:CCR}.
From now on, when we say ``CCR'' we mean a rectangle or an L-shaped CCR, with a single clipped corner.

\begin{figure*}[!ht]
	\centering
	\includegraphics[width=0.7\textwidth]{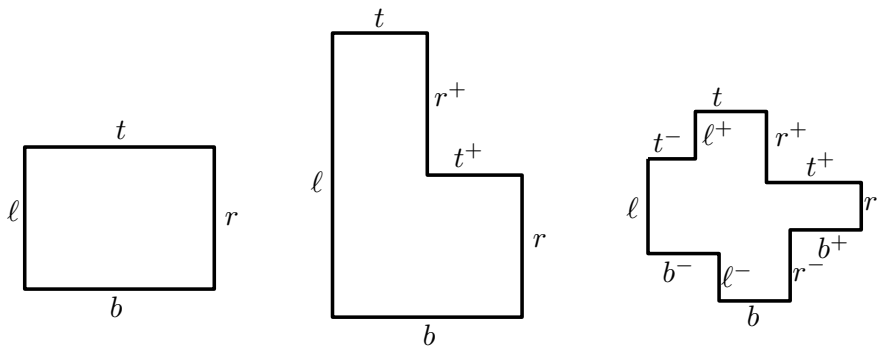}
	\caption{Examples of CCRs: Left: A rectangle with sides labeled (left $\ell$, right $r$, top $t$, and bottom $b$).
          Middle: An L-shaped CCR.
          Right: A general CCR with four clipped corners and 12 sides.}
	\label{fig:CCR}
\end{figure*}

\medskip
\noindent{\bf Maximal rectangles.}\quad
Let $I$ be an independent set of
rectangles within the bounding box $\B$.
We say that the set $I$ is {\em maximal within $\B$} if each rectangle
$R \in I$ has all four of its sides in nonzero-length contact with a side of another
rectangle of $I$, or with the boundary of $\B$.
An arbitrary independent subset $I=\{R_1,\ldots,R_k\}\subseteq {\cal
  R}$ of input rectangles $R_i$ can be converted into an independent
set of rectangles, $R'_i\supseteq R_i$, that are maximal within $B$,
in any of a number of ways, ``expanding'' each $R_i$
appropriately. (Note that the super-rectangles, $R'_i$ need not be
elements of the input set ${\cal R}$.)
To be specific, we define $I'=\{R'_1,\ldots,R'_k\}$ to be the set of
{\em maximal expansions}, $R'_i\supseteq R_i$, of the rectangles
$R_i$, obtained from the following process: Within the bounding box
$B$, we extend each rectangle $R_i$ upwards (in the $+y$ direction)
until it contacts another rectangle of $I$ or contacts the top edge of
$\B$. We then extend these rectangles leftwards, then downwards, and
then rightwards, in a similar fashion.  This results in a new set $I'$
of independent rectangles within $\B$, with each rectangle $R'_i$
containing the corresponding rectangle $R_i$, and each being maximal
within the set $I'$ of rectangles: each $R'_i$ is in contact, on all
four sides, either with another rectangle of $I'$, or with the
boundary of $\B$.  Note too that the sides of each $R'_i$ lie on the
grid ${\cal G}$ of horizontal/vertical lines through sides of the
input rectangles $R_i$; the coordinates defining the rectangles $I'$
are a subset of the set of coordinates of the input rectangles ${\cal R}$.
See Figure~\ref{fig:maximal}.

\begin{figure*}[!ht]
	\centering
	\includegraphics[width=0.35\textwidth]{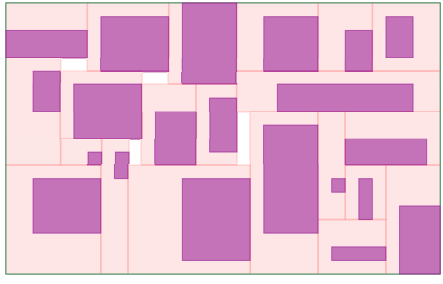}
	\caption{Example of the maximal expansions $R'_i$ (in light red) of
          a set $I$ of independent rectangles $R_i$ (in magenta).}
	\label{fig:maximal}
\end{figure*}

\medskip
\noindent{\bf The nesting relationship among maximal rectangles.}\quad
Consider a maximal rectangle, $R'_i\in I'$. By definition of
maximality, each of the four sides of $R'_i$ is in contact with the
boundary of $\B$ or with the boundary of another maximal rectangle,
$R'_j$. If $R'_i$ and $R'_j$ are in contact and the corresponding edge
of $R'_i$ lies within the {\em interior} of the corresponding edge of
$R'_j$, then we say that $R'_i$ is {\em nested} on that side where it
contacts $R'_j$. For example, if $R'_i$ is nested on its top side,
where it contacts the bottom side of $R'_j$, then the left side of
$R'_j$ is strictly left of the left side of $R'_i$, and the right side
of $R'_j$ is strictly right of the right side of $R'_i$.
For a rectangle $R'_i$ that is in contact with the boundary of $\B$,
we say that $R'_i$ is nested on a side that is contained in the
interior of a boundary edge of the bounding box $\B$.
We say that $R'_i$ is {\em nested vertically} if it is nested on its
top side or its bottom side (or both); similarly, we say that $R'_i$ is {\em
  nested horizontally} if it is nested on its left side or its right
side (or both).  Refer to Figure~\ref{fig:nested}. A simple observation that
comes from the definition, and the fact that rectangles in $I'$ do not
overlap, is the following:

%

\begin{observation}
  \label{obs:key}
For a set $I'$ of independent rectangles that are maximal within $\B$,
a rectangle $R'_i\in I'$ cannot be nested both vertically and horizontally.
\end{observation}

\begin{figure*}[!ht]
	\centering
	\includegraphics[width=0.2\textwidth]{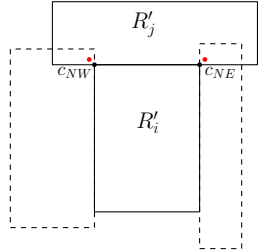}
	\caption{Proof of Observation~\ref{obs:key}: If $R'_i$ is
          nested vertically, with its top side being contained in the
          interior of the bottom side of $R'_j$, then the red points,
          which are slightly shifted top corners of $R'_i$, must be
          interior to $R'_j$. If $R'_i$ were also nested horizontally,
          abutting one of the dashed rectangles (of $I'$) to the
          left/right of $R'_i$, then we get a contradiction to the
          interior disjointness of rectangles, as a red point would be
          interior to two rectangles of $I'$.}
	\label{fig:key}
\end{figure*}

\begin{proof}
  Consider a rectangle $R'_i\in I'$ and assume that it is nested
  vertically, with its top side being a horizontal segment interior to
  the bottom side of a rectangle $R'_j\in I'$ that abuts $R'_i$ to its
  north. Let $c_{NE}=(x_{max},y_{max})$ be the northeast corner of
  $R'_i$; then, we know from the nesting relationship with $R'_j$, and the assumed integrality of
  coordinates defining the input rectangles, 
  that the shifted corner point
  $(x_{max}+1/2,y_{max}+1/2)$ is interior to $R'_j$. Similarly, the shifted corner point
  $(x_{min}-1/2,y_{max}+1/2)$, just northwest of the northwest
  corner, $c_{NW}$, of $R'_i$ is interior to $R'_j$.  However, if
  $R'_i$ were to be nested horizontally on its right side, then
  $(x_{max}+1/2,y_{max}+1/2)$ must be interior to the abutting
  rectangle to the right (and a similar statement holds if $R'_i$ were
  to be nested on its left side.) This a contradiction to the interior
  disjointness of the rectangles of~$I'$.
\end{proof}

\begin{figure*}[!ht]
	\centering
	\includegraphics[width=0.3\textwidth]{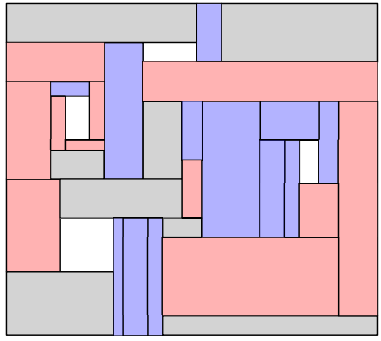}
	\caption{Example of maximal rectangles that are nested
          vertically (shown in blue) and nested horizontally (shown in
          red). Gray rectangles are not nested vertically or
          horizontally.}
	\label{fig:nested}
\end{figure*}

\medskip
\noindent{\bf CCR Partitions.}\quad
We define a recursive partitioning of a rectangle $B\subseteq \B$ as
follows.  At any stage of the decomposition, we have a partitioning of
$B$ into polygonal faces ({\em CCR faces}), each of which is a CCR
(meaning a rectangle or an L-shaped CCR). (Initially, there is a
single rectangular face,~$B$.)
A set, $\chi$, of axis-parallel line segments (called {\em horizontal
  cut segments} and {\em vertical cut segments}) on the grid ${\cal
  G}$ that partitions a CCR face $f$ into two or more CCR faces is
said to be a {\em CCR-cut} of $f$.  We often use $\sigma$ to refer to
a vertical cut segment of $\chi$, and we will consider $\sigma$ to be
a {\em maximal} segment that is contained in the union of grid edges
that make up $\chi$.
A CCR-cut $\chi$ is $K$-ary if it partitions $f$ into at most $K$ CCR
faces.  Necessarily, a $K$-ary CCR-cut $\chi$ consists of $O(K)$
horizontal/vertical cut segments, since each of the resulting $K$
subfaces are (constant-complexity, orthogonal) CCR faces.
By recursively making $K$-ary CCR-cuts, we obtain a $K$-ary hierarchical
partitioning of $B$ into CCR faces.  We refer to the resulting
(finite) recursive partitioning as a $K$-ary {\em CCR-partition} of $B$.
Associated with a CCR-partition is a {\em partition tree}, whose nodes
correspond to the CCR faces, at each stage of the recursive
partitioning, and whose edges connect a node for a face $f$ to each of
the nodes corresponding to the subfaces of $f$ into which $f$ is
partitioned by a $K$-ary CCR-cut.  The root of the tree corresponds to
$B$; at each level of the tree, a subset of the nodes at that level
have their corresponding CCR faces partitioned by a CCR-cut, into at
most $K$ subfaces that are children of the respective nodes.  The CCR
faces that correspond to leaves of the partition tree are called {\em
  leaf faces}.
We say that a CCR-partition is {\em perfect} with respect to a set of
rectangles if none of the rectangles have their interiors penetrated
by the CCR-cuts of the CCR-partition, and each leaf face of the
CCR-partition contains exactly one rectangle.  (This terminology
parallels the related notion of a ``perfect BSP'', mentioned below.)
We say that a CCR-partition is {\em nearly perfect} with respect to a
set of rectangles if for every CCR-cut in the hierarchical
partitioning each of the $O(1)$ cut segments of the CCR-cut penetrates
at most 2 rectangles of the set, and each leaf face of the
CCR-partition contains at most one rectangle.
(Having this property will enable a dynamic programming algorithm to
optimize the size of a set of rectangles, over all nearly perfect
partitions of the set, since there is only a constant amount of
information to store with each subproblem defined by a CCR face in the
CCR-partition.)

\medskip
\noindent{\bf Comparison with BSPs.}\quad
A binary space partition (BSP) is a recursive partitioning of a set of
objects using hyperplanes (lines in 2D); a BSP is said to be {\em
  perfect} if none of the objects are cut by the hyperplanes defining
the partitioning~\cite{de1997perfect}.  An orthogonal BSP uses cuts
that are orthogonal to the coordinate axes. There are sets of
axis-aligned rectangular objects for which there is no perfect
orthogonal BSP; see Figure~\ref{fig:bsp-vs-ccr}.
A CCR-partition is a generalization of an orthogonal binary space
partition (BSP), allowing more general shape faces (and cuts) in the
recursive partitioning, and allowing more than 2 (but still a constant
number of) children of internal nodes in the partition tree.

\begin{conjecture}\label{conj}
[Pach-Tardos~\cite{pach2000cutting}] For any set of $n$ interior-disjoint axis-aligned rectangles in the plane, there exists a subset of size $\Omega(n)$ that has a perfect orthogonal BSP.
\end{conjecture}

\begin{figure*}[!ht]
	\centering
	\includegraphics[width=0.3\textwidth]{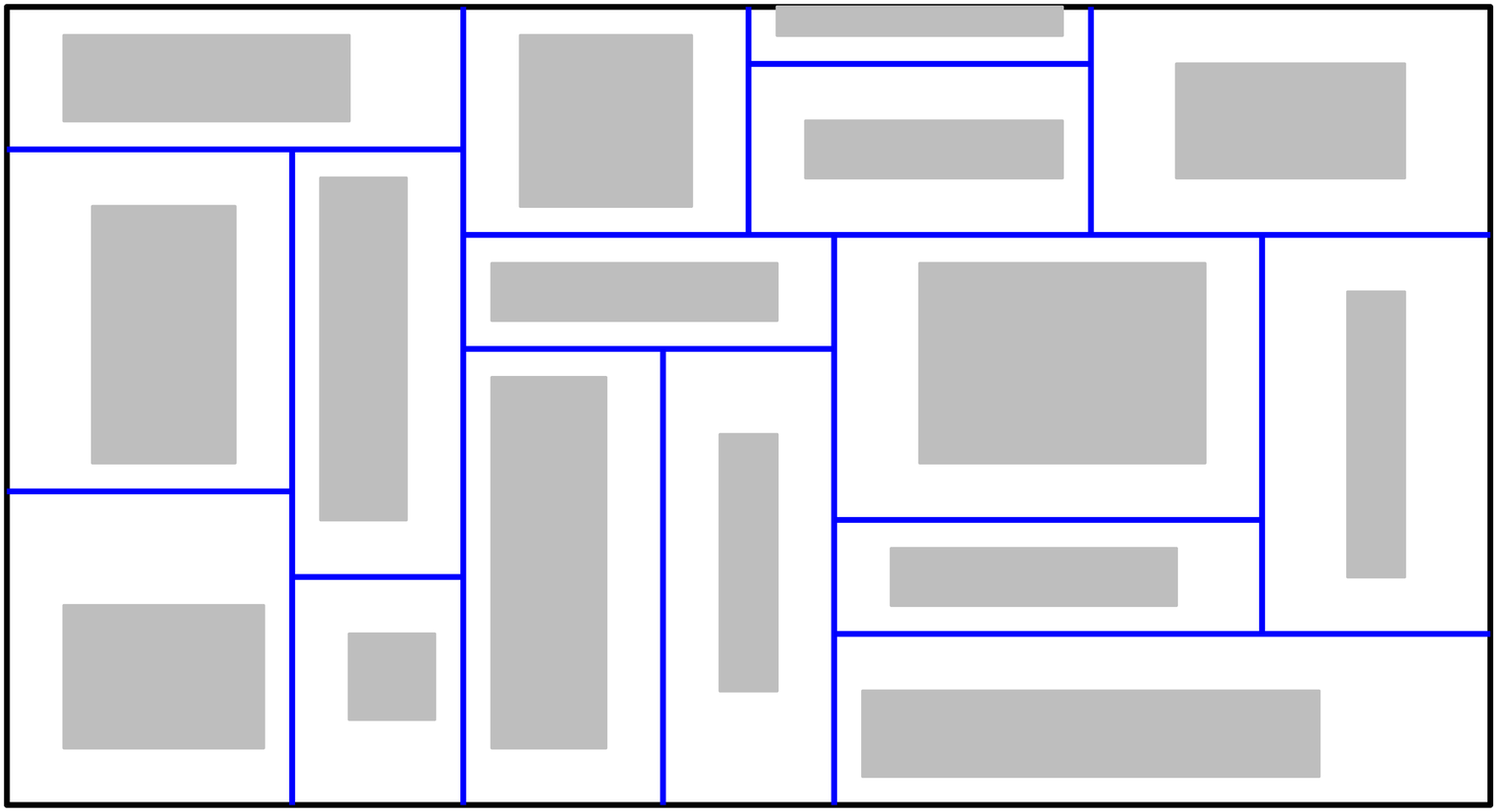}\hfill
       	\includegraphics[width=0.3\textwidth]{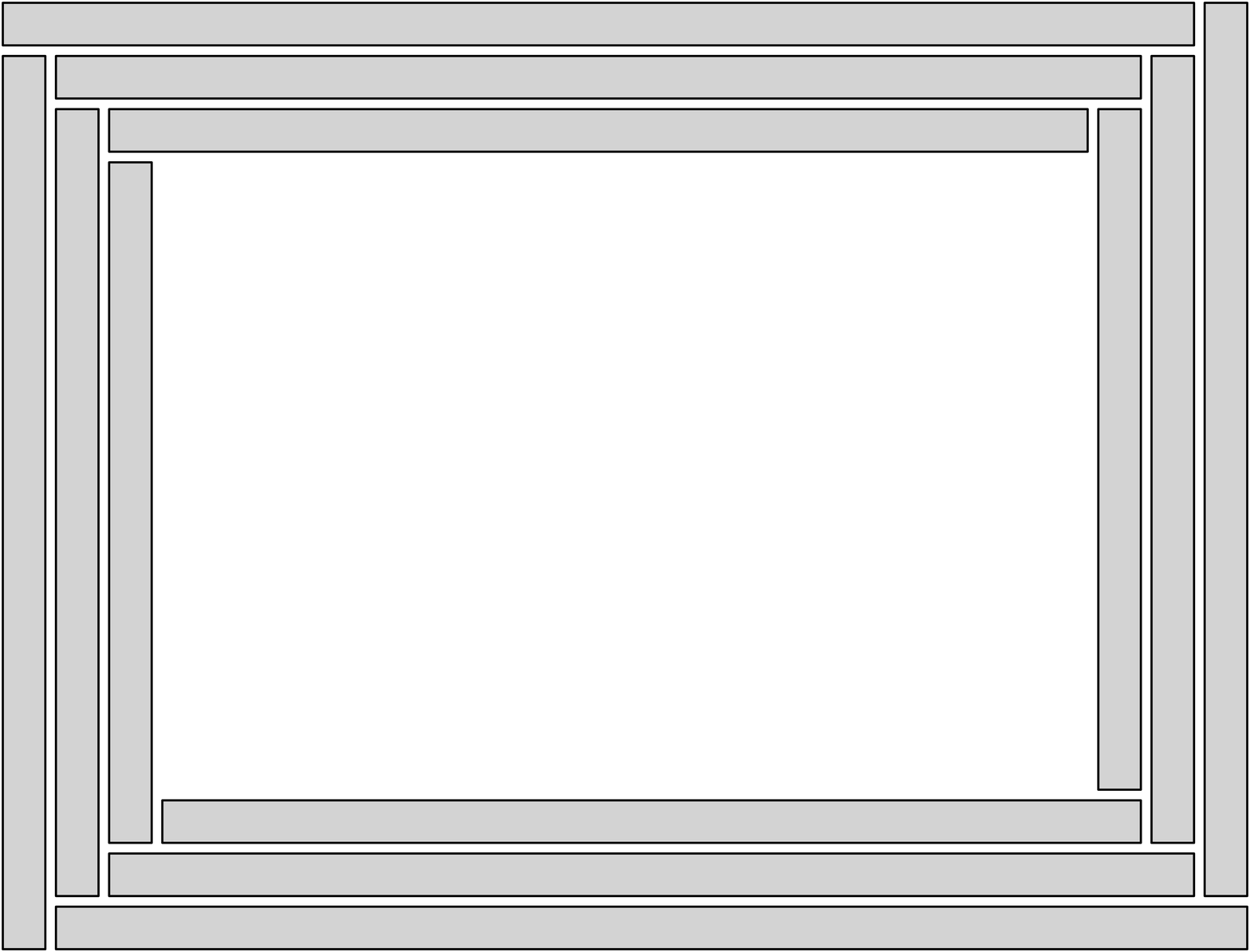}\hfill
        \includegraphics[width=0.3\textwidth]{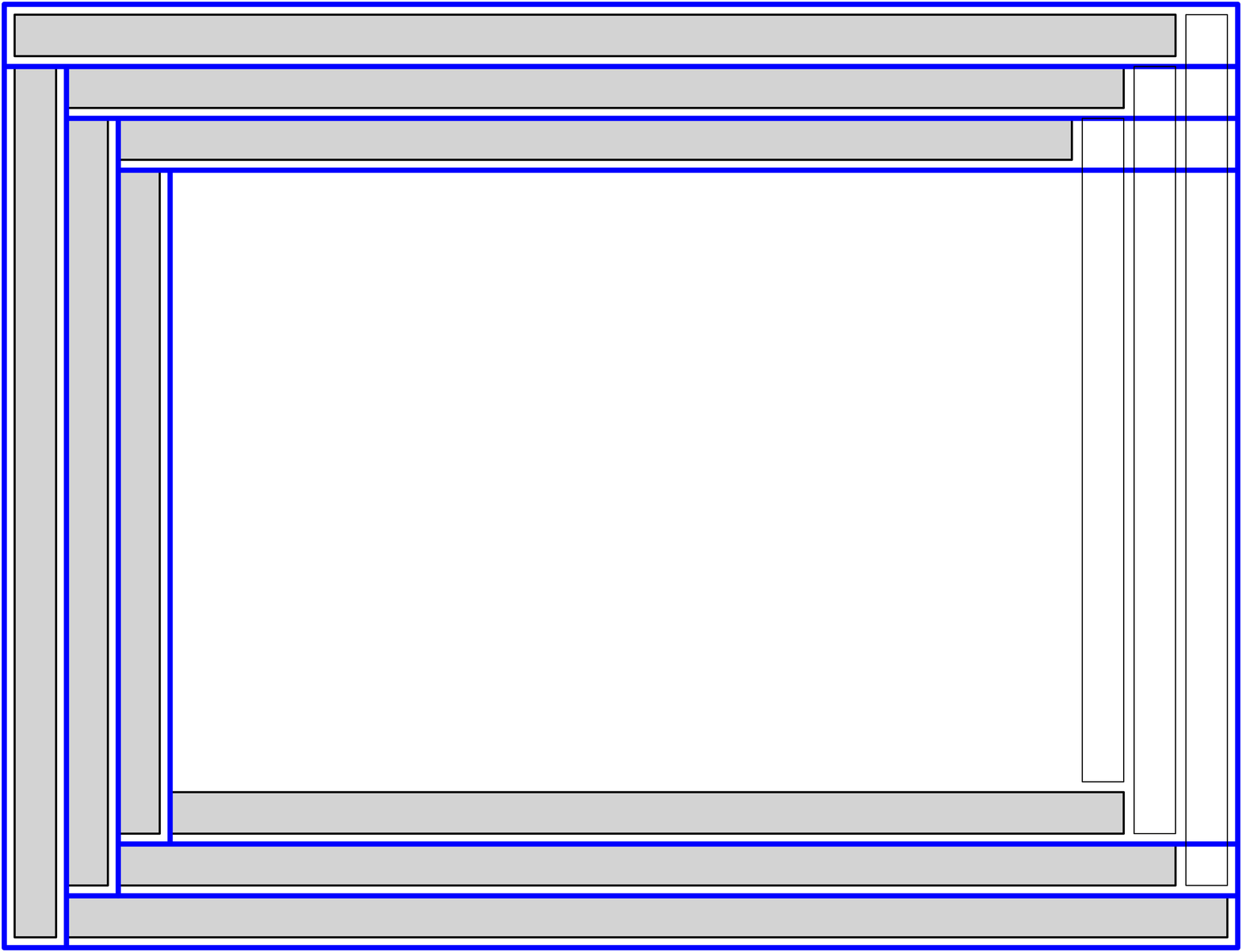}
	\caption{Left: A set of (gray) rectangles and a BSP that respects it. Middle: A set of (gray) rectangles that does not have a perfect BSP (but does have a perfect CCR-partition).
          Right: A perfect BSP partition of 3/4 of the rectangles in the middle figure (with 1/4 of the rectangles, in white, removed).}
	\label{fig:bsp-vs-ccr}
\end{figure*}

\section{The Structure Theorem}

Our main structure theorem states that for any set of $k$ interior disjoint
rectangles in the plane there is a constant fraction size 
subset with respect to which there exists a nearly perfect CCR-partition.
We state the theorem here and we utilize it for the algorithm in the next section;
the proof of the structure theorem is given in Section~\ref{sec:proof}.

While Conjecture~\ref{conj} remains open, and if true would imply a
constant-factor approximation for MISR (via dynamic programming), our
structural result suffices for obtaining a constant factor
approximation (Theorem~\ref{thm:main}), which is based on optimizing
over recursively defined CCR-partitions using dynamic programming.

\begin{theorem}\label{thm:structure}
  For any set $I=\{R_1,\ldots,R_k\}$ of $k$ interior disjoint
  (axis-aligned) rectangles in the plane within a bounding box $B$,
  there exists a $K$-ary CCR-partition of the bounding box $B$, with
  $K\leq 3$, recursively cutting $B$ into rectangles and (L-shaped) corner-clipped rectangles
  (CCRs), such that the CCR-partition is nearly perfect with
  respect to a subset of $I$ of size $\Omega(k)$. 
\end{theorem}

\section{The Algorithm}

The algorithm utilizes dynamic programming.  A subproblem ${\cal S}$
of the DP is specified by giving (1) a CCR $Q$ whose edges lie on the
grid ${\cal G}$ defined by the set of vertical/horizontal lines
through the sides of the input rectangles, and (2) a set $I_{\cal
  S}\subseteq {\cal R}$ of $O(1)$ input rectangles (these are
``special'' rectangles that we explicitly specify) that are penetrated
by some of the edges of $Q$, at most 2 per edge. 
Since $Q$ is either a rectangle or an L-shaped CCR, $Q$ has at most 6
edges/vertices, and we know that $|I_{\cal S}|\leq 12$ (a more careful count that exploits the structure of the cuts used in our CCR-partitions yields $|I_{\cal S}|\leq 6$) and that there are
only a polynomial number of possible CCRs.

For a subproblem ${\cal S}=(Q,I_{\cal S})$, let ${\cal R}({\cal
  S})\subseteq {\cal R}$ denote the subset of the input rectangles
that lie within $Q$ and are pairwise disjoint from the specified
rectangles $I_{\cal S}$.
Let $f({\cal S})$ denote the cardinality of a maximum-cardinality subset,
$I^*({\cal S})$, of ${\cal R}({\cal S})$ that is independent and for which there
exists a CCR-partition of $Q$ that is nearly perfect with respect to $I^*({\cal S})$.

The objective for the subproblem ${\cal S}$ specified by $(Q,I_{\cal S})$
is to compute a CCR-partitioning of $Q$ to maximize the
cardinality of the subset of ${\cal R}({\cal S})$ with respect to
which there is a CCR-partitioning of $Q$ that is nearly perfect.

The optimization is done by iterating over all candidate $K$-ary
($K\leq 3$) CCR-cuts of $Q$, together with choices of special
rectangles associated with the vertical cut segments of a CCR-cut,
resulting in the following recursion:
\begin{equation}
  f({\cal S})=\begin{cases}
0 & \mbox{if ${\cal R}({\cal S})=\emptyset$},\\
\max_{\chi\in\gamma({\cal S}), I_\chi} (f({\cal S}_1)+\cdots +f({\cal S}_K) + |I_\chi|) & \mbox{otherwise},
\end{cases}
\end{equation}
where the optimization is over the set, $\gamma({\cal S})$ of eligible
CCR-cuts within the grid ${\cal G}$ that partition $Q$ into $K$ CCRs
($2\leq K \leq 3$) and over the choices of the set $I_\chi \subseteq
I$ of $O(1)$ special rectangles, with at most 2 rectangles of $I_\chi$
penetrated by each cut segment of $\chi$.  The cut
$\chi\in \gamma({\cal S})$, together with the choice of $I_\chi$, results in a
set of $K$ subproblems ${\cal S}_1,\ldots,{\cal S}_K$, each specified
by a CCR, together with a set of specified special rectangles (a
subset of the $O(1)$ rectangles of $I_\chi$).  Since the eligible cuts
are of constant complexity within the grid of $O(n)$
vertical/horizontal lines, $|\gamma({\cal S})|$ is of polynomial
size. Also, there are a polynomial number of choices for the set
$I_\chi$, since each set is of constant size.  Thus, the algorithm runs in polynomial time.

The desired solution is given by evaluating $f({\cal S}_0)$, on the
subproblem ${\cal S}_0=(\B,{\cal R})$ defined by the full bounding
box, $\B$, of the input set ${\cal R}$ of rectangles.

The proof of correctness is based on mathematical induction on the
number of input rectangles of ${\cal R}$ that are contained in
subproblem ${\cal S}$. The base case is determined by subproblems
${\cal S}$ for which there is no rectangle of ${\cal R}$ inside the
CCR; we correctly set $f({\cal S})=0$ for such subproblems. With the
induction hypothesis that $f({\cal S})$ is correctly computed for
subproblems ${\cal S}$ containing at most $k$ rectangles of ${\cal
  R}$, we see that, for a subproblem ${\cal S}$ having $k+1$
rectangles of ${\cal R}$, since we optimize over all CCR-cuts
partitioning $Q$ into at most 3 CCRs, respecting $I_{\cal S}$, and
over all choices of $O(1)$ special rectangles penetrated by each
considered CCR-cut, with each of the values $f({\cal S}_i)$ being
correctly computed (by the induction hypothesis), we obtain a correct
value $f({\cal S})$. Since the algorithm runs in polynomial time, and
our structure theorem shows that there exists an independent subset of
${\cal R}$ of size at least $\Omega(k^*)$ 
for which there is a
CCR-partition respecting this subset, we obtain our main result:

\begin{theorem}\label{thm:main}
  There is a polynomial-time $O(1)$-approximation  
  algorithm for maximum
  independent set for a set of axis-aligned rectangles in the plane.
\end{theorem}

\noindent {\em Remark.} While efficiency is not the focus of our
attention here, other than to establish polynomial-time, we comment
briefly on the overall running time of the dynamic program. A naive
upper bound comes from the fact that there are at most $O(n^{6})$
L-shaped CCRs defined on the set of vertical/horizontal lines through sides of
the input rectangles, and there are at most $O(n^{6})$ specified rectangles $I_{\cal S}$ in any subproblem.
The most complex cut $\chi$ among the cases in the proof of Lemma~\ref{lem:key} is seen to have 3 horizontal cut segments and 1 vertical cut segment, with at most 2 associated specified rectangles along the vertical cut segment and at most additional specified rectangle along each of the horizontal cut segments; thus, there are at most $O(n^{9})$ choices for $(\chi,I_\chi)$. 
Thus, the
optimization is over at most $O(n^{9})$ choices, for each of at most $O(n^{12})$ subproblems.
The overall time bound of $O(n^{21})$ is surely an over-estimate.

\section{Proof of the Structure Theorem}
\label{sec:proof}

In this section we prove the main structure theorem,
Theorem~\ref{thm:structure}, which states that for any set
$I=\{R_1,\ldots,R_k\}$ of $k$ interior disjoint (axis-aligned)
rectangles in the plane within a bounding box $B$, there exists a
$K$-ary CCR-partition of the bounding box $B$, with $K\leq 3$,
recursively cutting $B$ into corner-clipped rectangles (CCRs), such
that the CCR-partition is nearly perfect with respect to a subset of
$I$ of size $\Omega(k)$ (specifically, we show size at least $k/4$).

  First, we consider the set $I'$ of maximal expansions, as described
  previously, of the $k$ input rectangles $I=\{R_1,\ldots,R_k\}$
  within $B$.  Recall that, appealing to Observation~\ref{obs:key},
  the set $I'$ is the disjoint union of subsets: $I'=I_h \cup I_v \cup
  I_0$, where $I_h$ is the subset of ``red'' rectangles that are
  nested horizontally, $I_v$ is the subset of ``blue'' rectangles that
  are nested vertically, and $I_0$ is the subset of ``gray''
  rectangles that are not nested on any of their four sides.  Refer to
  Figure~\ref{fig:nested}.  It cannot be that both $|I_h|$ and $|I_v|$
  are greater than $k/2$; thus, assume, without loss of generality,
  that $|I_h|\leq k/2$. Then, there are $|I_v|+|I_0| \geq k/2$
  ``non-red'' rectangles.\footnote{Throughout the discussion going
    forward in this section, the assumption that the number of red
    (maximal) rectangles is at most $k/2$ will permeate, leading to
    apparent ``bias'' in places where we construct ``fence'' segments
    that are horizontal, require constraints on vertical cut segments
    in Lemma~\ref{lem:key}, focus on charging off rectangles that are
    crossed by vertical cut segments, while horizontal cut segments
    never cross rectangles of $I'$, etc; in all of these places, there is a
    symmetric statement corresponding to the case in which, instead,
    the number of blue (maximal) rectangles is at most $k/2$, and
    there are at least $k/2$ ``non-blue'' rectangles.}

We now describe a process by which a subset of $\Omega(k)$ of the
rectangles of $I'$ is chosen, and a CCR-partition, $\Pi$, is
constructed that is nearly perfect with respect to this subset, and,
thus, by the following claim, is nearly perfect with respect to the
corresponding subset of (sub-)rectangles~$I$.

\begin{claim}\label{claim:1}
  If a CCR-partition is nearly perfect with respect to a set
  $A'\subset I'$ of maximal rectangles, then it is nearly perfect with
  respect to the corresponding set $A\subseteq I$ of input rectangles.
\end{claim}

\begin{proof}
  This follows immediately from the fact that, for each $R'\in I'$,
  the (unique) corresponding $R\in I$ is a subrectangle of $R'$.
\end{proof}

Initially, all rectangles of $I'$ are {\em active}.
During the process of recursively partitioning $\B$ into a
CCR-partition, a subset of the rectangles of $I'$ will be discarded
and removed from active status.  In the end, the set of remaining
active rectangles of $I'$ have the property that the CCR-partition is
nearly perfect with respect to the remaining active rectangles.  We
show (Claim~\ref{fact7}), via a charging scheme, that at most a
constant fraction (3/4) of the rectangles of $I'$ are discarded,
implying that the final set of active rectangles of $I'$ has
cardinality at least $k/4=|I'|/4$.
%

At each stage in the recursive partitioning, a CCR $Q$ that contains
more than one rectangle of $I'$ (and thus more than one rectangle of
$I$) is partitioned via a cut $\chi$, within the grid ${\cal G}$, into
at most $K=3$ CCR faces, via a set of $O(1)$ horizontal and vertical
cut segments. Details of how we determine these CCR-cuts are presented
in the statement and proof of Lemma~\ref{lem:key}, below, after we
introduce the notion of ``fences''.

A property of the cuts $\chi$ we construct is that no rectangle of
$I'$ is crossed by a horizontal cut segment (Lemma~\ref{lem:key},
(iii), together with the property that the horizontal segments
(``fences'') that serve as input to the lemma are defined to have the
property that they do not cross any of the rectangles of $I'$);
horizontal cut segments can follow the boundary of a rectangle of
$I'$, and might penetrate up to two rectangles of $I'$ (namely, the rectangle(s) 
that contain each of the two endpoints of the segment), but cannot cross
any such rectangle.

A vertical cut segment, $\sigma$, of $\chi$ can penetrate and also
{\em cross} some rectangles of $I'$: we discard those rectangles that
are {\em crossed} by such a cut segment $\sigma$. On any vertical cut
segment at most two rectangles (the rectangle(s) containing each of
the two endpoints of $\sigma$) can be {\em penetrated} but not {\em
  crossed}; these rectangles will not be discarded and are the source
of the use of ``nearly perfect'' cuts in our method.  The ability to
find such cuts $\chi$ is assured by Lemma~\ref{lem:key}.  Those
rectangles that are crossed by a vertical cut segment and discarded
must be accounted for, to assure that a significant fraction remain in
the end.  This is where a charging scheme is utilized, which we
describe below, after we introduce fences.

\medskip
\noindent{\bf Fences.}\quad
An important aspect of the method we will describe for construction of
a CCR-partition that respects a large subset of $I'$ is the
enforcement of constraints that prevent the cuts we produce from
crossing rectangles of $I'$ without being able to ``pay'' for the crossing,
through our charging scheme (below).  We achieve this by establishing
certain {\em fences}, which are horizontal line segments, on the grid
${\cal G}$, each ``anchored'' on a left/right side of the current CCR
face, $Q$. The fences become constraints on CCR-cuts of $Q$, as we
require that a CCR-cut $\chi$ have no vertical cut segments that cross
any fence (and all horizontal cut segments will be contained within the
set of fence segments); refer to Lemma~\ref{lem:key} below.

For each of the 4 corners, $c_{NE}$, $c_{NW}$, $c_{SW}$, $c_{SE}$, of
a rectangle $R\in I'$, we let the corresponding {\em shifted} corner
(denoted $c'_{NE}$, $c'_{NW}$, $c'_{SW}$, $c'_{SE}$) be the point
interior to $R$, shifted towards the interior of $R$ by $1/2$ in both
$x$- and $y$-coordinates (recall that the defining coordinates of the
input rectangles are assumed, without loss of generality, to be
integers); for example, if $c_{NE}=(x,y)$, then $c'_{NE}=(x-1/2, y-1/2)$.
(There is nothing special about the number 1/2; any real
number strictly between 0 and 1 could be used to define shifted
corners for the rectangles, which have integer coordinates.)

We say that a corner, $c$, of a rectangle $R\in I'$, is
{\em exposed to the right within $Q$} if a horizontal rightwards ray,
$\rho$, with endpoint at its corresponding shifted corner, $c'$, does
not penetrate any rectangle of $I'$ (other than $R$) that lies fully
within $Q$ before reaching the right boundary of $Q$.
If a corner, $c$, of rectangle $R$ is
{\em not} exposed to the right within $Q$, then the rightwards ray
$\rho$ penetrates at least one rectangle that lies fully within $Q$
before reaching the right boundary of $Q$; the first such rectangle,
$R'$, that is penetrated by $\rho$ after $\rho$ exits $R$ is called
the {\em right-neighbor of $R$ at $c$}.
We similarly define the notion of a
corner $c$ being exposed to the left within $Q$ and the notion
of a left-neighbor of $R$ at $c$ (in the case that the corner $c$ is not exposed to the left). Note that, from the definition, the
northeast corner of $R$ is exposed to the right (resp., left) if and
only if the northwest corner of $R$ is exposed to the right (resp.,
left); a similar statement holds for southeast and southwest corners.
Thus, we say that the top side of $R$ is exposed to the
left (resp., right) if the northeast and northwest corners of $R$ are exposed
to the left (resp., right).
If the top side of $R$ is not exposed to the left (resp., right), then we speak of the top-left-neighbor of $R$ (resp., top-right-neighbor of $R$), which is the left-neighbor (resp., right-neighbor) of $R$ at a top corner of $R$.
Similarly, we speak of a bottom side of $R$ being exposed to the left/right, and, if not exposed, having a bottom-left-neighbor or bottom-right-neighbor. We use the term {\em left-neighbor} of a rectangle to refer to a top-left-neighbor or a bottom-left-neighbor; we similarly use the term {\em right-neighbor}.
%

When a CCR $Q$ is partitioned by a CCR-cut, $\chi$, a vertical cut
segment, $\sigma$, of $\chi$ may cause some top/bottom sides of other
rectangles of $I'$ within the new CCR faces (subfaces of $Q$ on either
side of $\sigma$) to become exposed to the left or to the right within
their respective new CCR faces: a leftwards/rightwards ray (from a shifted corner of
some rectangle) that previously penetrated a rectangle of $I'$ (within $Q$) now reaches $\sigma$,
a new vertical boundary of a subface, before the penetration.
When a vertical cut segment $\sigma$ of a CCR-cut causes the
top/bottom side of a rectangle $R\in I'$ that lies fully within a
subface to become exposed to its left (resp., right) within its new
CCR face, we establish a {\em fence} (horizontal line segment) that
contains the top/bottom side of $R$ and extends to the left (resp.,
right) boundary of the CCR face (at a point on $\sigma$ that gave rise
to the new exposure of $R$) and extends to the right (resp., left)
side of the corresponding right-neighbor (resp., left-neighbor) of the
rectangle~$R$.\footnote{In an earlier version of this
  paper~\cite{DBLP:journals/corr/abs-2101-00326}, the fence segment
  extended only to the right (or left) side of the newly exposed
  rectangle $R$, rather than continuing to the right (resp., left) side of
  the corresponding right-neighbor (resp., left-neighbor) of $R$.  By extending the fence through the left/right-neighbor,
  we are able to ensure (Claim~\ref{fact5}) that no rectangle has both
  its northwest and northeast corners charged (or both the southwest
  and southeast corners), as the extended fence ensures that the
  left/right-neighbor is not crossed by a cut.} See
Figure~\ref{fig:fence-extended}.

\begin{figure*}[!ht]
	\centering
	\includegraphics[width=0.7\textwidth]{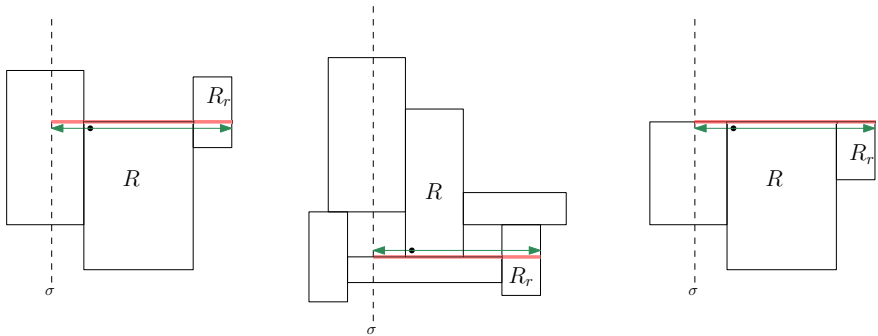}
	\caption{Examples showing the establishment of fences after a
          cut along vertical segment $\sigma$ causes the northwest
          corner (left and right cases) or the southwest corner
          (middle case) of $R$ to become exposed to the left. The newly
          established fence segment is shown in red, extending from
          $\sigma$ to the right side of $R_r$, the top-right-neighbor
          (left and right cases) or bottom-right-neighbor (middle
          case).}
        \label{fig:fence-extended}
\end{figure*}

Note that, by its definition, a fence segment does not cross any rectangle
of $I'$; it may penetrate a (left/right-neighbor) rectangle of $R\in I'$ near one of the two ends of the fence segment.
(Since an input rectangle can lie within the interior of its corresponding maximal rectangle in $I'$, this means that
the fence segment may cross an input rectangle that is contained
within one of the at most two penetrated maximal rectangles of $I'$.)
Refer to Figure~\ref{fig:fences}; we use the convention that fences
are shown in red (resp., blue) if anchored on their left (resp.,
right) endpoint.  (This color choice of red/blue for fences has
nothing to do with the ``red'' and ``blue'' color terminology we had
above for rectangles in the set $I_h$ and $I_v$.)
As cuts are made during the process, we establish fences in order to maintain the following
invariant during the course of the recursive partitioning:

\begin{quotation}
\noindent [Fence Invariant] {\em For each rectangle $R\in I'$ that
  lies fully within a CCR $Q$, if the top/bottom side of $R$ is exposed to the left
  (resp., right) within $Q$, then there is a (horizontal) fence
  segment $\alpha$ (resp., $\beta$) that includes the top/bottom side of $R$
  and extends leftwards to the left side of $Q$ and rightwards to the
  right side of the top/bottom-right-neighbor of $R$ (resp., extends rightwards
  to the right side of $Q$ and leftwards to the left side of the
  top/bottom-left-neighbor of $R$).}
  \end{quotation}

\begin{figure*}[!b]
	\centering
	\includegraphics[width=\textwidth]{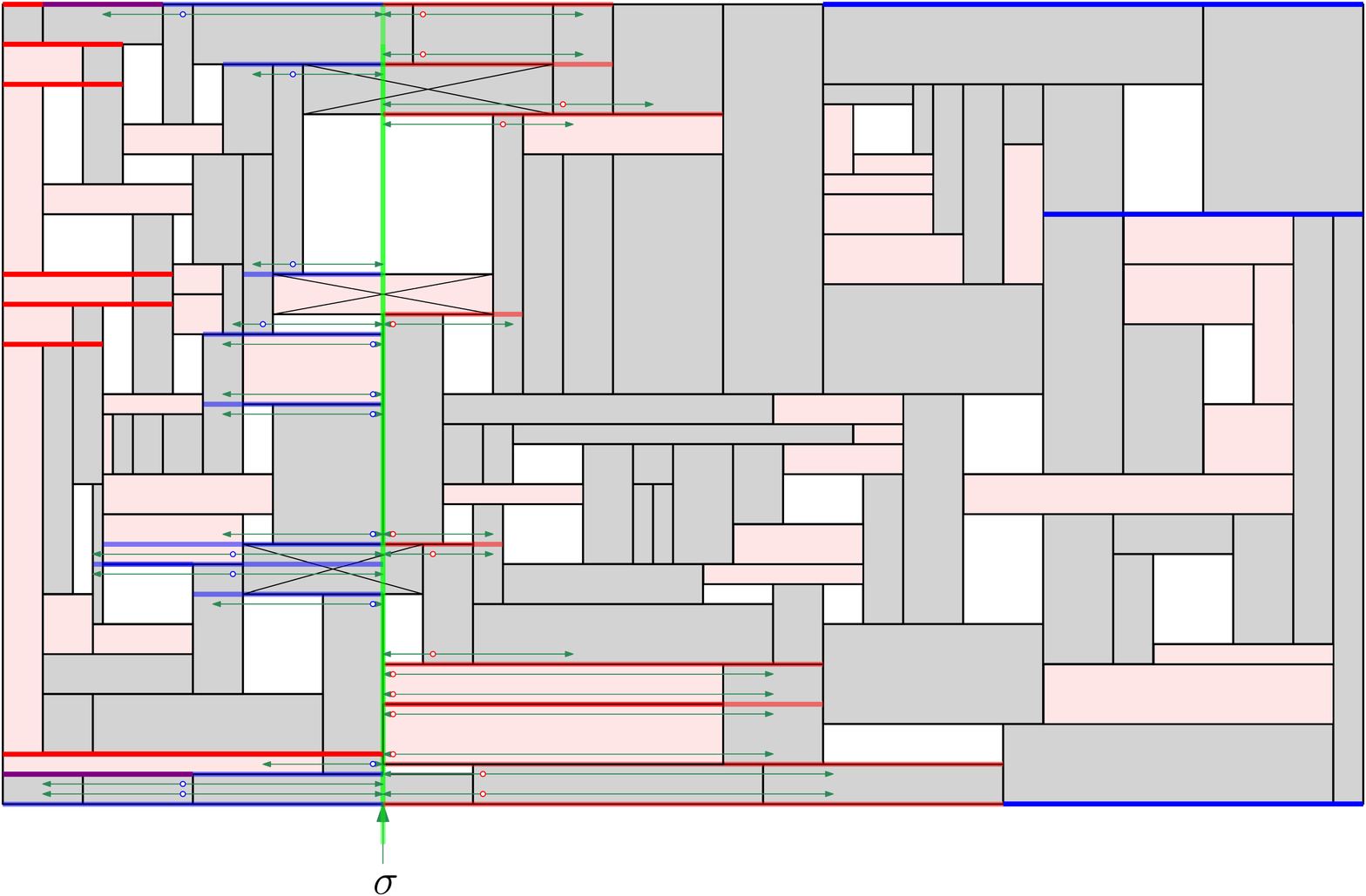}
	\caption{Example showing the establishment of fences after a
          cut. In this figure, the CCR is a rectangle, the maximal
          rectangles are shown, with red ones (nested to the left or
          to the right) shown in light red, and all non-red rectangles
          shown in gray. A vertical cut segment $\sigma$ (in green) is
          shown. The cut results in the top/bottom sides of some
          rectangles becoming exposed; we show the respective shifted
          corners as blue/red hollow dots, with dark green arrows
          showing the exposure to the cut $\sigma$, as well as showing
          the additional penetrated neighboring rectangle, which
          defines the extent of the fence segment. Rectangles that are
          {\em crossed} by the cut are marked with an ``X'': (Note
          that the first and last rectangles intersected by $\sigma$
          are penetrated but not crossed by $\sigma$.)  All of the
          newly established fences, anchored on $\sigma$, are shown in
          red (if the left endpoint is anchored on $\sigma$) or in
          blue (if the right endpoint is anchored on $\sigma$).}
	\label{fig:fences}
\end{figure*}

We say that a rectangle $R\in I'$, with $R$ fully contained within a CCR
$Q$, is {\em grounded} if its top or bottom side is
contained within the top or bottom boundary of $Q$ or within a fence
segment within $Q$, or if $R$ is penetrated (but not crossed) by a fence segment that extends from one vertical side of $R$ to the other.
Lemma~\ref{lem:key} guarantees that the CCR-cuts we utilize do
not have any vertical cut segments $\sigma$ that cross any fences and,
therefore, do not cross any grounded rectangles (since a vertical
segment that crosses a rectangle necessarily crosses its top and
bottom sides and any horizontal segment that penetrates the rectangle, extending from one vertical side to the other). It can be, however, that a vertical cut segment
$\sigma$ penetrates (without crossing) a grounded rectangle, e.g., if
the segment $\sigma$ terminates on a fence or on a top/bottom side of
$Q$; however, a vertical cut segment $\sigma$ can penetrate, without
crossing, at most two grounded rectangles, one containing the top
endpoint of $\sigma$ and one containing the bottom endpoint of $\sigma$. (This is
the source of the ``2'' in the definition of a nearly perfect cut.)
Each rectangle $R\in I'$ is a maximal expansion of an original input
rectangle of $I$; since the original rectangle is a subrectangle of
$R$, it can happen that a vertical cut segment $\sigma$ crosses the
original rectangle (or it could completely miss intersecting it)
associated with a grounded rectangle $R\in I'$, while not {\em
  crossing}, but only {\em penetrating},~$R$.

Lemma~\ref{lem:key} below shows that the recursive partitioning
process is feasible, showing that for any CCR face $Q$ and any set of
(horizontal) fences anchored on the left/right sides of $Q$, it is
always possible to find a CCR-cut with the properties claimed in the
lemma; in particular, the cut should ``respect'' the given set of
fences, in that no vertical cut segment of the cut crosses a fence
segment.  We initialize the set of fences in order that the
Fence Invariant holds at the beginning of the process:
for each rectangle $R\in I'$ whose top/bottom side is exposed to the left
within $B$, we establish a fence along its top/bottom side, anchored on the
left side of $B$, extending rightwards to the right side of the
(top/bottom)-right-neighbor of $R$. We do similarly
for rectangles whose top/bottom side is exposed to the right within $B$.  See
Figure~\ref{fig:nested-fences}.
(If the top/bottom side of $R$ is contained within the
top/bottom side of $B$, we need not establish a fence (though we do show them in the figure), since the fence
segment would be contained in the top/bottom side of $B$.) If there exists a rectangle $R\in I'$
whose top/bottom side is exposed both to the left and to the right within
$B$, then we simply cut $B$ with a horizontal cut segment along the
top/bottom side of $R$, from the left side of $B$ to the right side of $B$,
since this horizontal segment penetrates no rectangles of~$I'$.

\begin{figure*}[!ht]
	\centering
	\includegraphics[width=0.35\textwidth]{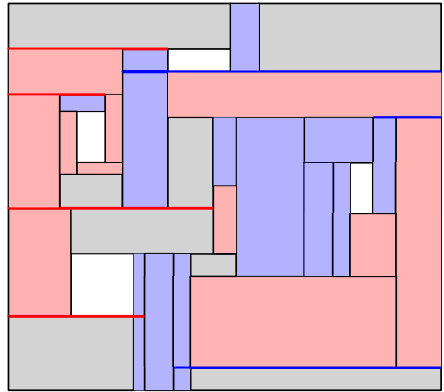}
	\caption{The initial fences, shown as bold red/blue segments, for the shown set of maximal rectangles within a rectangle $B$.}   
	\label{fig:nested-fences}
\end{figure*}

We appeal to the following technical lemma, whose proof is
based on a simple enumeration of cases, given in the appendix.

  \begin{lemma}\label{lem:key}
    Let $Q$ be a CCR whose edges
    lie on a grid lines of ${\cal G}$.
    Let $\{\alpha_1,\ldots,\alpha_{k_\alpha}\}$ be a set of $k_\alpha\geq 0$ ``red''
    horizontal anchored (grid) segments within $Q$, that are anchored
    with left endpoints on the left side of $Q$, and let $\{\beta_1,\ldots,\beta_{k_\beta}\}$ be
    a set of $k_\beta\geq 0$``blue'' horizontal anchored (grid) segments within $Q$,
    that are anchored with right endpoints on the right side of $Q$.  Then, assuming that $Q$ contains at
    least two grid cells (faces of ${\cal G}$), there exists a CCR-cut
    $\chi$ with the following properties:
    \begin{description}
    \item[(i)] $\chi$ partitions $Q$ into $O(1)$ (at most 3) CCR faces;
    \item[(ii)] $\chi$ is comprised of $O(1)$ horizontal/vertical
      segments on the grid ${\cal G}$, with endpoints on the grid;
    \item[(iii)] horizontal cut segments of $\chi$ are a subset of the
      given red/blue anchored segments;
    \item[(iv)] vertical cut segments of $\chi$ do not cross any of
      the given red/blue anchored segments;
    \item[(v)] there is at most 1 vertical cut segment of $\chi$.
      \end{description}
  \end{lemma}

The lemma shows that a CCR, $Q$, with anchored (grid) segments has a
CCR-cut, on the grid, partitioning $Q$ into at most 3 CCRs (CCR
faces in the grid). Applying this recursively (and finitely, due to
the finiteness of the grid), until there is at most one rectangle of $I'$
within each CCR face, we will be able to conclude, based on the 
charging scheme below, the proof of the claimed structure theorem.

\medskip

\noindent {\em Remark.} There is an equivalent ``rotated'' version of
Lemma~\ref{lem:key} that applies to a set of vertical (fence) segments
anchored on the top/bottom sides of $Q$, with vertical cut segments
lying along the anchored segments, and at most 1 horizontal cut segment not
crossing any anchored (vertical) segments. The rotated version applies in our
proof of the structure theorem in the case that the non-blue maximal
rectangles have cardinality at least $k/2$.

\medskip
\noindent{\bf Charging Scheme.}\quad
Consider a vertical cut segment, $\sigma$, that is part of a CCR-cut
$\chi$ of a CCR $Q$, given by Lemma~\ref{lem:key}.  The segment
$\sigma$ penetrates some (possibly empty) subset, $I'_\sigma$, of
rectangles of $I'$ that lie within $Q$; the intersection of $\sigma$
with the rectangles $I'_\sigma$ is a set of (interior-disjoint)
subsegments of $\sigma$, at most two of which share an endpoint with
$\sigma$ and are thus {\em not} contained within the interior of
$\sigma$. Thus, at most two of the rectangles of $I'_\sigma$ are not
{\em crossed} by $\sigma$.

\begin{claim}
  \label{fact1}
 If a rectangle of $I'_\sigma$ is crossed by $\sigma$, then neither
 its top nor its bottom is exposed to the left or right within $Q$.
\end{claim}
\begin{proof}
This is immediate from the Fence Invariant, and the fact
(Lemma~\ref{lem:key}, part (iv)) that vertical cut segments do not
cross fences.
\end{proof}

Our charging scheme assigns (``charges off'') each {\em non-red}
rectangle $R$ that is crossed by a vertical cut segment $\sigma$ to
two corners of rectangles of $I'$ (one corner to the left of $R$, one
corner to the right of $R$), each receiving 1/2 unit of charge.
If the rectangle $R$ is {\em red}, we do not attempt to charge it off,
as our charging scheme is based on charging to corners of rectangles
that lie to the left/right of $R$ (within the vertical extent of $R$),
and a red rectangle is, by definition, nesting on its left or on its
right (or both), implying that we will not have available the corners
we need to be able to charge. Instead, our accounting scheme below
will take advantage of the fact that the number of non-red rectangles
is at least~$k/2$.

Consider a non-red rectangle $R\in I'$ within $Q$ that is crossed by a
vertical cut segment $\sigma$.  By property (v) of
Lemma~\ref{lem:key}, there are no other
vertical cut segments of $\chi$.
We charge the crossing of $R$ to 2 corners, each receiving 1/2 unit of charge to ``pay'' for the crossing of $R$:
we charge 1/2 to a corner of a rectangle of $I'$
that lies to the right of $R$ (within the vertical extent of $R$),
and 1/2 to a corner of a rectangle of $I'$ that lies to the left of $R$ (within the vertical extent of $R$).

We describe and illustrate the assignment of charge to a corner that lies to the right of $R$;
the assignment to a corner that lies to the left of $R$ is done symmetrically.
Specifically, let $c$ be the northeast corner of $R$, and let $c'$ be the
corresponding shifted corner (slightly to the southwest of $c$). By
Claim~\ref{fact1}, $c$ is not exposed to the right within $Q$; thus,
the rightwards horizontal ray from $c'$ must penetrate at least one
rectangle of $I'$ that is fully within $Q$ before reaching the right boundary of
$Q$. Let $R_r$ be the first (leftmost) rectangle that is penetrated ($R_r$ is the top-right-neighbor of $R$).
Let $y^+$ and $y^-$ (resp., $y^+_r$ and $y^-_r$) be the top and bottom
$y$-coordinates of rectangle $R$ (resp., $R_r$). Similarly, let $x^+$,
$x^-$, $x^+_r$, and $x^-_r$ denote the $x$-coordinates of the
left/right sides of $R$ and $R_r$, so that
$R=[x^-,x^+]\times[y^-,y^+]$ and
$R_r=[x_r^-,x_r^+]\times[y_r^-,y_r^+]$.  By definition of $R_r$, we
know that $y^+_r \geq y^+$ and that $x^+\leq x^-_r$.
Refer to Figure~\ref{fig:charging-case-analysis}, where we illustrate the
charging scheme (for charging to the right of $R$), enumerating the 8 potential cases, depending if (i)
$y^+=y^+_r$ or $y^+<y^+_r$, (ii) $y^->y^-_r$ or $y^-\leq y^-_r$, (iii)
$x^+=x^-_r$ or $x^+<x^-_r$.
We itemize the 8 cases below:

\begin{figure*}[!ht]
	\centering
	\includegraphics[width=\textwidth]{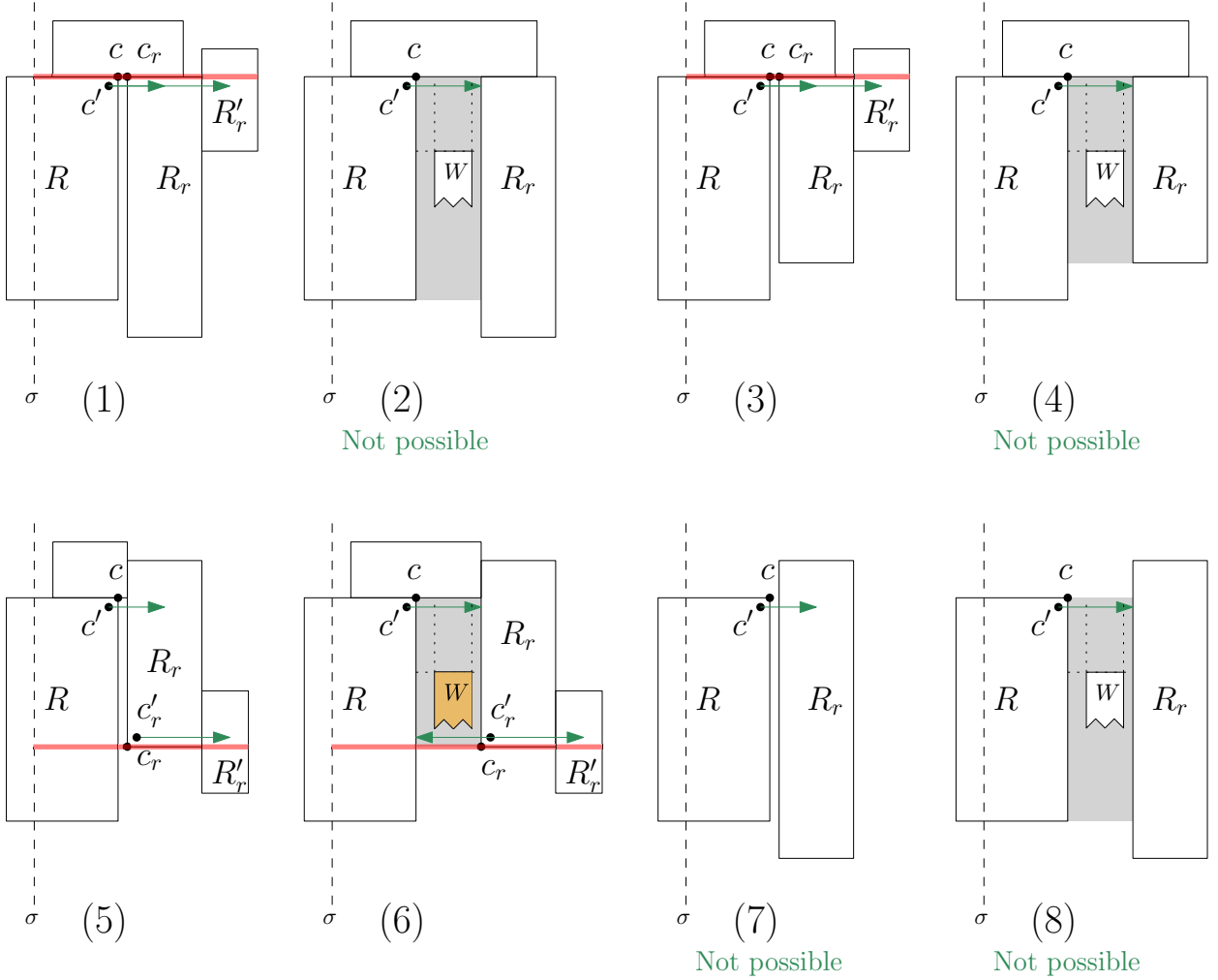}
	\caption{Case analysis for the charging scheme when a non-red
          rectangle $R$ is crossed by a vertical portion (shown
          dashed), $\sigma$, of a cut $\chi$.  The charged corner is
          labelled $c_r$, for the possible cases (1), (3), (5),
          (6). The fences established after the cut are shown in
          red. Here, we only show the case of charging to the right of
          the crossed rectangle~$R$, assigning a charge of 1/2 to an appropriately selected corner; symmetrically, we also charge 1/2 to a corner that is to the left of~$R$. 
          .}
	\label{fig:charging-case-analysis}
\end{figure*}

\begin{description}
\item[(1)] [$y^+=y^+_r$, $y^-> y^-_r$, and $x^+=x^-_r$.] We charge the
  northwest corner, $c_r$, of rectangle $R_r$.

  After the cut along $\sigma$, rectangle $R_r$ will have its top side
  exposed to the left, so, to maintain the Fence Invariant, we
  establish a horizontal fence segment (shown in red), anchored on
  $\sigma$, passing through the top side of $R_r$, and extending
  rightwards to the right side of the top-right-neighbor, $R'_r$, of
  $R_r$.  (We know that the top side of $R_r$ is not exposed to the
  right side of $Q$, since, if it were exposed, then, by the Fence
  Invariant, at the moment of first exposure there would be a fence
  extending along the top side of $R_r$, and leftwards through $R$
  (the top-left-neighbor of $R_r$), to the left side of $R$, implying
  that this fence segment is crossed by $\sigma$, a contradiction.)
  This fence will become exposed to the left (at a point along
  $\sigma$) after the cut along $\sigma$.  In the figure, we show the
  case in which $R'_r$ is abutting $R_r$; the case in which there is a
  gap between $R_r$ and its top-right-neighbor $R'_r$ is similar.

\item[(2)] [$y^+=y^+_r$, $y^-> y^-_r$, and $x^+<x^-_r$.]  In this
  case, there is a rectangular region (shown in gray in the figure)
  separating the right side of $R$ from the left side of $R_r$. This
  gray rectangle must intersect the interior of at least one rectangle
  of $I'$; otherwise, $R$ is not maximal, as it can be extended to the
  right, until it abuts $R_r$. Thus, there must be some rectangle in
  $I'$ whose interior overlaps with the interior of the gray
  rectangle; let $W\in I'$ be such a rectangle whose top side has the
  maximum $y$-coordinate among all such rectangles. Then, we get a
  contradiction as follows: If the top side of $W$ lies at or above
  the top sides of $R$ and $R_r$, then the rightwards ray from $c'$
  penetrates $W$ before penetrating $R_r$, contradicting the
  definition of $R_r$; on the other hand, if the top side of $W$ lies
  below the rightwards ray from $c'$, then we get a contradiction to
  the maximality of $W$, since it can be extended upwards, to the top
  of the gray rectangle.  Thus, this case cannot happen.
 
\item[(3)] [$y^+=y^+_r$, $y^-\leq y^-_r$, and $x^+=x^-_r$.] This is
  handled exactly as in case (1): we charge the northwest corner $c_r$
  and establish the fence segment (shown in red), through the top side
  of $R_r$, extending to the right side of the top-right-neighbor, $R'_r$.

\item[(4)] [$y^+=y^+_r$, $y^-\leq y^-_r$, and $x^+<x^-_r$.] Exactly as
  in case (2), this case cannot happen.

\item[(5)] [$y^+<y^+_r$, $y^-\leq y^-_r$, and $x^+=x^-_r$.]  We charge
  the southwest corner, $c_r$, of rectangle $R_r$. After the cut along
  $\sigma$, rectangle $R_r$ will have its bottom side exposed to the
  left, so, to maintain the Fence Invariant, we establish a horizontal
 fence segment (shown in red), anchored on
  $\sigma$, passing through the bottom side of $R_r$, and extending
  rightwards to the right side of the bottom-right-neighbor, $R'_r$, of
  $R_r$.  (We know that the bottom side of $R_r$ is not exposed to the
  right side of $Q$, since, if it were exposed, then, by the Fence
  Invariant, at the moment of first exposure there would be a fence
  extending along the bottom side of $R_r$, and leftwards through $R$
  (the bottom-left-neighbor of $R_r$), to the left side of $R$, implying
  that this fence segment is crossed by $\sigma$, a contradiction.)

\item[(6)] [$y^+<y^+_r$, $y^-\leq y^-_r$, and $x^+<x^-_r$.] We charge
  the southwest corner, $c_r$, of rectangle $R_r$. Consider the gray
  rectangle, $[x^+,x^-_r]\times[y^-_r,y^+]$; we claim that there can
  be no rectangle $W\in I'$ whose interior overlaps with this gray
  rectangle, for, if such a rectangle existed, then one such
  rectangle, $W$, that maximizes the $y$-coordinate of the top side
  fails to be maximal. Thus, no rectangle fully within the subface to the right of $\sigma$ is penetrated by the
  leftwards ray from $c_r'$ before it reaches $\sigma$, on the left side of the subface. Thus, after
  the cut along $\sigma$, rectangle $R_r$ will have its bottom side
  exposed to the left. In order to maintain the Fence Invariant, we
  establish a horizontal 
 fence segment (shown in red), anchored on
  $\sigma$, passing through the bottom side of $R_r$, and extending
  rightwards to the right side of the bottom-right-neighbor, $R'_r$, of
  $R_r$.  (As in case (5), we know that there must be a bottom-right-neighbor,
  since the bottom side of $R_r$ cannot have been exposed to the right, as this would imply
  the existence of a fence across $R$, contradicting the fact that the cut segment $\sigma$ crosses no fence segment.)
  led to a fence that

\item[(7)] [$y^+<y^+_r$, $y^-> y^-_r$, and $x^+=x^-_r$.] We get a
  contradiction to the fact that $R$ is non-red, since the
  inequalities $y^+<y^+_r$ and $y^-> y^-_r$ say that $R$ is nesting on
  its right. Thus, this case cannot happen.

\item[(8)] [$y^+<y^+_r$, $y^-> y^-_r$, and $x^+<x^-_r$.] By reasoning
  exactly as in case (2), this case cannot happen.
\end{description}

Summarizing the above case enumeration of the charging scheme, we
state the following facts:

\begin{claim}\label{fact2}
A corner of a rectangle is charged at most once (with a charge of 1/2).
\end{claim}

\begin{proof}
  When a corner $c_r$ is charged, as illustrated in the case
  enumeration of Figure~\ref{fig:charging-case-analysis} (cases (1),
  (3), (5), (6)), the corner $c_r$ becomes exposed to the left and a
  horizontal fence segment is established (according to the Fence
  Invariant), extending from the segment $\sigma$, along the top/bottom edge
  of the rectangle whose corner $c_r$ is being charged, and rightwards
  to the right side of the top/bottom-right-neighbor of that rectangle.  Once
  exposed to the left, a corner remains exposed to the left in the CCR
  face containing it, for the remainder of the process of recursive
  partitioning.
\end{proof}

\begin{claim}
  \label{fact3}
  If a corner, $c_r$, of rectangle
  $R_r\in I'$ is charged, then $R_r$ has not previously been crossed
  by a vertical cut segment, is not currently (as part of $\chi$)
  crossed by a vertical cut segment, and will not be crossed by
  vertical cut segments later in the recursive CCR-partition.
\end{claim}

\begin{proof}
  In all cases (1), (3), (5) and (6), if $R_r$ had previously been crossed
  by a vertical cut segment, $\sigma'$, then the northeast corner $c$
  of $R$ would have become exposed to the right (since $R_r$ is the
  top-right-neighbor of $R$), and a fence segment would have been
  established through $c$, along the top side of $R$ (extending
  leftwards to the left side of the top-left-neighbor of $R$); this
  fence prevents $R$ from being crossed, since a vertical cut segment
  never crosses a fence segment.

  There is no other vertical cut segment of $\chi$ (recall property (v) of Lemma~\ref{lem:key}), besides $\sigma$; thus, $R_r$ is
  not crossed by a vertical cut segment of $\chi$.

  Later in the recursive CCR-partition, the rectangle $R_r$ cannot be
  crossed by a vertical cut segment, since such a crossing would imply
  a crossing of the fence we established through the top or bottom 
  side of $R_r$ (and this is ruled out by property (iv) of
  Lemma~\ref{lem:key}). Indeed, we establish fences in order to enforce
  that once a rectangle has a corner charged, it cannot be crossed by
  later cuts in the recursive partitioning.
\end{proof}

\begin{claim}
  \label{fact5}
  At most two corners of a rectangle are charged: it is never the case
  that both a northwest and a northeast corner of the same rectangle
  is charged (and a similar statement holds for southwest and southeast
  corners).
\end{claim}

\begin{proof}
  If, as in cases (1) and (3), the northwest corner, $c_r$, of a rectangle
  $R_r\in I'$ is charged for the crossing of a rectangle $R$, then a fence is established along the top side of $R_r$,
  extending rightwards to the right side of the top-right-neighbor, $R'_r$, of
  $R_r$. This fence prevents $R'_r$ from being
  crossed later by a vertical cut segment; thus, the northeast corner of
  $R_r$ cannot be charged, since, if the
  northeast corner were charged for the crossing of some rectangle
  $R'$ to its right, then $R'$ must be the top-right-neighbor of $R_r$ (just as we saw in cases (1) and (3) that when the northwest corner $c_r$ is charged for the crossing of $R$ by $\sigma$, the rectangle $R$ must be the top-left-neighbor of $R_r$), and we know
  that the top-right-neighbor of $R_r$ cannot be crossed by a vertical cut segment, since it
  is protected by the previously established fence.  A similar statement holds for cases (5) and (6), when the southwest corner of $R_r$ is charged for the crossing of $R$: it cannot be that the southeast corner is also charged, as the bottom-right-neighbor of $R_r$ is protected by the fence segment established.
\end{proof}

The set of red rectangles can be partitioned into two sets: those that
are cut (let $h_\chi$ be the number) and those that are not cut (let
$h_0$ be the number).
Similarly, non-red rectangles can be partitioned into two sets: those that are
cut (let $m_\chi$ be the number) and those that are not cut (let $m_0$
be the number).
By Claim~\ref{fact3}, if any rectangle is charged,
then it is never crossed by a cut; thus, there are at most $h_0+m_0$
rectangles that are charged, and we know from
Claim~\ref{fact2} that the amount of charge received by any of these
rectangles is at most 1 (the crossing of a rectangle $R$ by $\sigma$ resulted in a charge of 1/2 to a northwest corner and 1/2 to a northeast corner).
We conclude that

\begin{claim}
  \label{fact6}
  The total sum of all charges is at most $h_0+m_0$, and thus $m_\chi\leq h_0+m_0$. 
\end{claim}

\begin{proof}
  By Claim~\ref{fact5}, at most two corners of a rectangle is charged (each with a charge of 1/2), and by Claim~\ref{fact2}, each is charged at most once.
Since, by specification of the charging scheme, there is exactly one unit of charge per crossed non-red rectangle, we get that $m_\chi\leq h_0+m_0$.
\end{proof}

In the recursive partitioning, we discard every one of the $h_\chi+m_\chi$ rectangles that are crossed by vertical cut segments;
what remains in the end are the $h_0+m_0$ rectangles that were not crossed and therefore not discarded. (These remaining rectangles either appear as isolated rectangles within a leaf face of the CCR-partition, or appear as special rectangles, which are penetrated but not crossed.)
%
%
Now, Claim~\ref{fact6}, together with the fact that $m_0+m_\chi\geq
k/2$ (by assumption), imply that
$$h_0+m_0 \geq m_\chi \geq (k/2) - m_0,$$
from which we see that
$$2h_0 + 4m_0 \geq k,$$
and thus, 
$$4(h_0+m_0) \geq k.$$
This implies that
$$h_0+m_0\geq k/4,$$
which proves our main claim (with a 4-approximation),
\begin{claim}\label{fact7}
  The total number, $h_0+m_0$, of rectangles that are not crossed by
  vertical cut segments (and thus that survive) in the recursive
  partitioning is $h_0+m_0 \geq k/4$.
\end{claim}

\subsection*{Improving the Approximation Factor}

While our main goal in this paper is simply to get an
$O(1)$-approximation algorithm that runs in polynomial time, we
briefly discuss how our methods yield an improvement of the
approximation factor 4 described above, if some additional care is
taken in the charging scheme.

Consider the cases ((1), (3), (5) and (6)), where we described the
charging of 1/2 to a corner $c_r$ (northwest or southwest of $R_r$,
the top-right-neighbor of $R$), to account for half of the cost of
$\sigma$ crossing $R$ (the other half being charged, by a symmetric
argument, to a corner to the left of $R$).
We now impose a preference to charge a northwest corner, if possible,
of a rectangle for which $R$ is its top-left-neighbor.  In cases (1)
and (3), $c_r$ is already a northwest corner. In cases (5) and (6),
provided that the bottom of $R_r$ is not aligned with the bottom of
$R$ (i.e., provided that $y_r^- > y^-$), there must be (by maximality)
at least one rectangle abutting $R_r$ from below, so we let $\bar R_r$
be that rectangle abutting $R_r$ with leftmost left side (so that $R$ is the top-left-neighbor of $\bar R_r$); instead of
the southwest corner of $R_r$, we can charge the northwest corner of
$\bar R_r$, and the top of $\bar R_r$ will become exposed after the
cut along $\sigma$, so we establish a fence, anchored on $\sigma$,
containing the top side of $\bar R_r$, and extending rightwards to the
right side of the top-right-neighbor of $\bar R_r$.  In cases (5) and
(6), if $y_r^- = y^-$, then we go ahead and charge the southwest
corner of $R_r$, as we do not have the option to charge the northwest
corner of $\bar R_r$, which abuts $R_r$ from below, since the interior
of $\bar R_r$ lies strictly below the bottom side of $R$ (implying
that $R$ is not the top-left-neighbor of $\bar R_r$).

Now, consider a rectangle $R_r$ for which both of its left corners
(northwest and southwest) is charged (an amount 1/2 each).  As the
above argument shows, this happened because the northwest corner of
$R_r$ was charged (in case (1) or (3)), since the top side of $R_r$ is
aligned with the top of a crossed (by $\sigma$) rectangle $R_1$, and
its bottom side was also aligned with the bottom side of a crossed
rectangle $R_2$.  (Further, from maximality, we see that the right
sides of $R_1$ and $R_2$ are abutting the left side of $R_r$.)

There are two cases: (A) $R_r$ is red (in which case it must be nesting on its right, since it is not nesting on its left); and
(B) $R_r$ is non-red.
In case (A), we keep the full charge of 1/2 assigned to each of the two left corners of $R_r$;
this is accounted for in the term $h_0$ of the total charge, since $R_r$ is then one of the $h_0$ red rectangles that is
not crossed.
In case (B), we know (by the same cases analysis, cases (1)-(8), which showed the existence of a chargeable corner to the right of $R$) that there must be a chargeable corner, $c_{rr}$, on the
left side of a rectangle $R_{rr}$ that is the top-right-neighbor of $R_r$. We reallocate the two charges (of 1/2 each) to the left
corners of $R_r$: instead, we assign charges of 1/3 to each of the two left corners of $R_r$, and assign 1/3 to the corner $c_{rr}$.
To enable this charging scheme, we must modify slightly our specification of fences, so that a fence segment extends not just to the right side of the 
(immediate) top/bottom-right-neighbor of a rectangle whose top/bottom becomes exposed to the left, but so that the fence extension continues through the
{\em next} rectangle penetrated to the right. This extended fence prevents $R_{rr}$ from having more than two of its corners charged (each with 1/3),
as it will not be possible to charge both northwest and northeast (or to charge both southwest and southeast) corners, just as we saw in the proof of Claim~\ref{fact5}.
With this extended fence, now penetrating two rectangles (one of which may be crossed) instead of one,  our notion of ``nearly perfect'' must allow for the possibility of at most 3 (still $O(1)$) rectangles penetrated per edge of
a CCR subproblem.  This, in turn, implies an increased (polynomial) running time in the dynamic program.

This modification of the charging yields an improved approximation
factor: Since each of the $h_0$ uncrossed red rectangles has at most two corners charged, each with 1/2, and each of the $m_0$ uncrossed non-red rectangles has at most a total charge of 2/3 (arising from either a single corner charge of 1/2, or two corners charged, each with 1/3), we get that the total charge, $m_\chi$, is at most $h_0 +
(2/3)m_0$, which shows, using $m_0+m_\chi\geq k/2$,
$$h_0+(2/3)m_0 \geq m_\chi \geq (k/2) - m_0,$$
from which we get
$$2h_0 + (10/3)m_0 \geq k,$$
implying that
$$h_0+m_0\geq 3k/10,$$
yielding a 10/3-approximation.

This method of offloading charge to a next ``layer'' of neighboring
rectangles can be continued to additional layers, yielding better and
better factors, approaching a factor of 3. For example, if $R_{rr}$
has both of its left corners charged (each with an offloaded charge of
1/3), as a result of there being two rectangles ($R_{r,1}$ and
$R_{r,2}$) each having both of their left corners charged (due to a
sequence, $R_1$, $R_2$, $R_3$, $R_4$, of consecutive rectangles being
crosssed by $\sigma$), then, if $R_{rr}$ is non-red, we offload to the
corner of a right-neighbor $R_{rrr}$, while extending fences into the
next neighboring rectangle; this results in a factor 22/7 (the 4
charges of 1/2 get distributed as charges of 2/7 to each of the two
left corners of $R_{r,1}$ and $R_{r,2}$, the two left corners of
$R_{rr}$, and one left corner of $R_{rrr}$, so that each non-red
uncrossed rectangle receives charge at most 4/7). The running time
goes up, though, with each application of this method, as there become
more special rectangles needed to specify within the state of the
dynamic program, as we allow fences to penetrate more neighboring
layers. It appears that this approach to obtaining an improvement is
``stuck'' at a factor $3+\epsilon$, so that additional ideas may be
needed to improve the charging scheme further.

\section{Conclusion}

Our main result is a first polynomial-time constant-factor
approximation algorithm for maximum independent set for axis-aligned
rectangles in the plane.  Our goal here was to devise new methods that
achieve a constant factor; can the methods be developed further to
yield an improved constant approximation factor?

In the time since the original
version~\cite{DBLP:journals/corr/abs-2101-00326} appeared, where an
initial approximation factor of 10 was given, Galvez et
al~\cite{galvez20214approximation} obtained an improved factor of 6
via a novel variant of methods from
\cite{DBLP:journals/corr/abs-2101-00326} (using a fairly simple, less
technical argument), and a further improvement to factor 4 (using more
involved technical arguments). The same authors have a further
  improvement to factor 3~\cite{galvez2021-3approximation}, which
  requires considerably more technical effort and new ideas.
%
%

More ambitiously, is there a PTAS? Our approach seems to give up at
least a factor of 2 in its method of handling the nestedness issue, in
order to guarantee enough rectangles can be charged to nearby
rectangle corners.  Can the running time of a constant-factor
approximation algorithm be improved significantly? Can the Pach-Tardos
Conjecture~\ref{conj} be resolved?  (This would yield a
constant-factor approximation algorithm with an improved running
time.)

\appendix

\section*{Appendix}

\subsection*{Proof of Lemma~\ref{lem:key}}

\begin{proof}
    The proof is by a simple enumeration of cases, each of which is
    illustrated in Figure~\ref{fig:cases-new}.

In the case in which $k_\alpha=0$ or $k_\beta=0$ (there are no fence
segments, or only fence segments anchored on one side (left or right)
of $Q$), a cut consisting of a single vertical chord (within the grid
${\cal G}$) suffices to partition $Q$ into two subfaces, satisfying
the claim.  So, assume now that $k_\alpha, k_\beta\geq 1$.

In the case itemization below, we first consider the case (case (1)) in which $Q$ is a rectangle and
then consider the cases (cases (2) through (6)) in which $Q$ is an
L-shaped CCR, which, without loss of generality, has its northeast
corner clipped, as shown in the figures.
Cases (2) through (6) are distinguished based on how a particular rectangle, $C$, shown
shaded gray in Figure~\ref{fig:cases-new}, has its top and bottom sides among the fences and the top/bottom sides of $Q$.
The rectangle $C$ is one face (cell) in the {\em vertical decomposition} (also known as the
{\em trapezoidal decomposition}~\cite{BCKO})
of $Q$ with respect to the
fence segments: the vertical decomposition is a decomposition of $Q$ into rectangular faces
(cells) induced by the fence segments, together with vertical segments
through endpoints of fences (and the single reflex vertex, $v$, of $Q$) that
are maximal within $Q$ without crossing any fence segment.
Refer to Figure~\ref{fig:vert-decomp}.
Specifically, we let $C$ be the unique rectangular face containing the slightly shifted point $v'$ (shown as a small black dot in
the figures), which is just northwest of the integral reflex vertex $v=(v_x,v_y)$,
say, at $v'=(v_x-1/2,v_y+1/2)$.

\begin{figure*}[!htbp]
\centering
\includegraphics[width=0.5\textwidth]{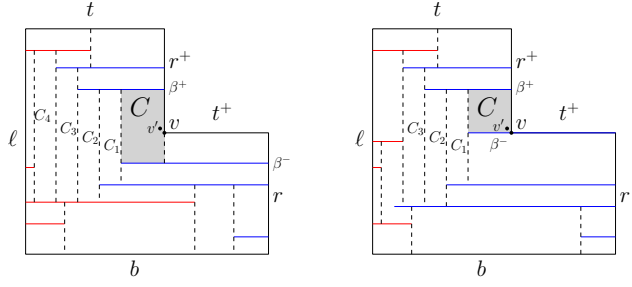}
\caption{The vertical decomposition of an L-shaped CCR $Q$ with
  respect to red and blue fence segments. The case shown is a CCR with
  the northeast corner clipped; the unique reflex vertex is $v$, and
  the slightly shifted, to the northwest, point $v'$ is shown as a
  small dot.  The cell $C$ is defined to be the unique (rectangular)
  face containing $v'$. Two examples are shown, both the the top and
  bottom sides of $C$ contained in blue fence segments $\beta^+$
  (anchored on $r^+$) and $\beta^-$ (which we consider to be anchored
  on the top endpoint of $r$, or equivalently the bottom endpoint of
  $r^+$).}
\label{fig:vert-decomp}
\end{figure*}

We partition into cases according to the top/bottom sides of the
rectangle $C$: the bottom of $C$ might lie on the bottom side, $b$, of
$Q$, or on a blue fence $\beta^-$, or on a red fence $\alpha$; the top
of $C$ might lie on the top side, $t$, of $Q$, or on a blue fence
$\beta^+$ anchored on the right side $r^+$ at a point strictly above
the reflex vertex $v$.  (The top side of $C$ cannot lie on a fence
that is not anchored on the upper right side, $r^+$, of $Q$, since $C$
is a face in the vertical decomposition that contains $v'$, and thus
abuts the right side segment $r^+$ of $Q$.)

We have the following cases:
\begin{description}
\item[(1)] [$Q$ is a rectangle] 
  Let $\beta$ be the blue fence with leftmost left endpoint.
    \begin{description}
    \item[(a)] [The rightmost red endpoint is left of the left endpoint
        of $\beta$] Then, a vertical chord through the left endpoint
      of $\beta$ does not cross any fence, so we cut along this
      vertical chord, and along $\beta$, obtaining 3 subfaces, each a
      rectangle.
    \item[(b)] [Otherwise] A vertical ray (e.g., upwards) from the
      left endpoint of $\beta$ meets a red fence, $\alpha$ before
      reaching the boundary of $Q$. We make a ``Z'' cut along $\beta$,
      a vertical segment to $\alpha$, and a portion of $\alpha$,
      obtaining two subfaces, each a CCR.
    \end{description}

    \item[(2)] [The top of $C$ lies on $t$, the bottom on $b$] We make
      a straight vertical cut along the right side of $C$,
      obtaining 2 rectangle subfaces.

    \item[(3)] [The top of $C$ lies on $t$, the bottom on blue fence
      $\beta$] Since $C$ contains $v'$, the fence $\beta$ must be
      anchored on the right side, $r$, of $Q$, at or below the top
      endpoint of $r$ (if anchored at the top endpoint of $r$, we can
      equivalently consider $\beta$ to be anchored at reflex vertex
      $v$).  We make an ``L'' cut along the left side of $C$ and a
      portion of $\beta$.  The left side of $C$ may have been
      determined by the left endpoint of $\beta$ (case (3)(a) in the
      figure), or it may have been determined by the right endpoint of
      a red fence (case (3)(b) in the figure), in which case our cut
      includes the red fence segment as well. Assuming (as shown in
      the figure) $\beta$ is anchored on a point of the right side $r$
      below vertex $v$, we obtain 2 subfaces that are L-shaped CCRs,
      and one rectangular subface (in case (3)(b)); if $\beta$ is
      anchored at reflex vertex $v$, the figures would be similar, but
      there would be 1 subface that is an L-shaped CCR, and one or two
      (in case (3)(b)) rectangular subfaces.

    \item[(4)] [The top of $C$ lies on blue fence $\beta$, the bottom
      on $b$] This case is like (3), and we obtain 2 subfaces that are
      CCRs, and possibly one rectangular subface (in case (4)(b)). (An
      alternative cut is to make a single vertical cut along the right
      side of $C$.)


      \begin{figure*}[!htbp]
\centering \includegraphics[width=\textwidth]{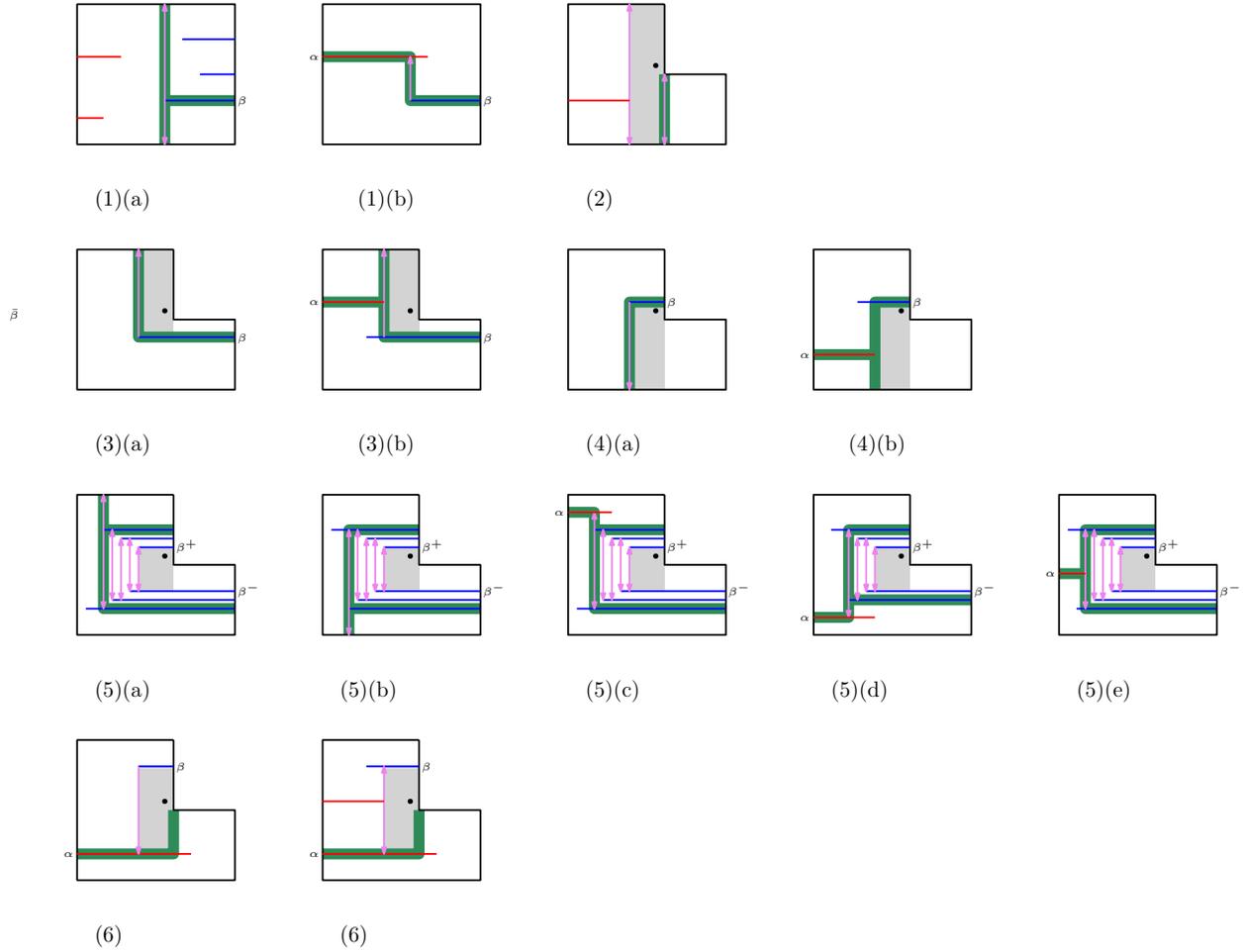}
\caption{The cases in the proof of Lemma~\ref{lem:key}.  The
  rectangular cell $C$ is shown shaded gray; it is the cell of the
  vertical decomposition that contains $v'$, the slightly shifted (to
  the northwest) reflex vertex $v$ of $Q$, shown as a small black
  dot. Cuts are shown in thick green. The fences are shown in blue if their right
    endpoint lies on the right boundary (sides $r^+$ or $r$)
    of $Q$, and are referred to using the letter beta (``$\beta$, $\beta^-$, $\beta^+$'');
    fences with left endpoint on the left boundary (side $\ell$)
    of $Q$ are shown in red and are referred to
    using a letter alpha (``$\alpha$'').  Note that in cases (3) and (5) the blue fence segment is shown anchored at a point interior to the right side, $r$; it could be anchored at the top endpoint of $r$, or equivalently at the reflex vertex $v$.
    Only those fence segments that are needed to identify the subcases
    are shown; there may be numerous other red/blue fences.
}
\label{fig:cases-new}
\end{figure*}

    \item[(5)] [The top of $C$ lies on blue fence $\beta^+$, the
      bottom on blue fence $\beta^-$] The left side of $C$ is a
      vertical segment with endpoints on the blue fences $\beta^+$ and
      $\beta^-$.  Informally, consider sweeping this segment
      leftwards, keeping it of maximal length, while not crossing any
      fence segment: the upper endpoint moves along a ``staircase''
      going leftwards and making jumps upwards, while the lower
      endpoint moves along a staircase going leftwards and making
      jumps downwards, until one of five events (enumerated as cases
      (a) through (e) below) takes place. More formally, starting at
      (rectangular) cell $C$ in the vertical decomposition, we follow
      a path leftwards in the dual graph of the decomposition, from
      cell to a unique left neighboring cell, until we come to a cell,
      $C'$ (which is possibly $C$ itself), that has 2 (or more) left
      neighbors (because its left side is determined by the right
      endpoint of a red fence), or until we reach a cell whose top and
      bottom sides are not both on blue fences.
%
      %
      %
      We have the following subcases depending on the classification of the cell $C'$:

  \begin{description}
  \item[(a)] [The top of $C'$ is on the top side, $t$, of $Q$] We obtain 2 subfaces
      that are L-shaped CCRs, and one rectangular subface.
  \item[(b)] [The bottom of $C'$ is on the bottom side, $b$, of $Q$] We obtain 2 subfaces
      that are L-shaped CCRs, and one rectangular subface.
  \item[(c)] [The top of $C'$ is on a red fence segment $\alpha$] We obtain 3 subfaces
      that are L-shaped CCRs.
  \item[(d)] [The bottom of $C'$ is on a red fence segment $\alpha$] We obtain 3 subfaces
      that are L-shaped CCRs.
  \item[(e)] [The left side of $C'$ has a right endpoint of a red fence $\alpha$] We obtain 3 subfaces
      that are L-shaped CCRs.
  \end{description}
  In Figure~\ref{fig:vert-decomp}, one example (on the left) has
  $C'=C_2$ (case (d), with the bottom of $C'$ on a red fence), and one
  example (on the right) has $C'=C_3$ (case (e), with the left side of
  $C'$ determined by the right endpoint of a red fence).

\item[(6)] [The top of $C$ is on a blue fence $\beta$, the bottom of $C$ is on a red fence $\alpha$]
  We obtain 2 subfaces: one rectangle and one L-shaped CCR.

    \end{description}

In all cases, we see that there is a single vertical cut segment, at
most 3 subfaces, and all portions of the cuts lie on the grid, with
horizontal portions along fence segments.  This concludes the proof of
the lemma.
\end{proof}



\subsection*{Acknowledgements}
I thank
Mathieu Mari for his input on an earlier draft~\cite{DBLP:journals/corr/abs-2101-00326}. I thank Arindam Khan and the other coauthors of
\cite{galvez20214approximation} for discussions and for sharing their preprint~\cite{galvez2021-3approximation}.

This work was partially supported by the National Science Foundation
(CCF-2007275), the US-Israel Binational Science Foundation (BSF
project 2016116), Sandia National Labs, and DARPA (Lagrange).

\bibliographystyle{plainurl}
\bibliography{refs}

\end{document}